\documentclass{article}

\usepackage{siunitx}
\usepackage{float}
\usepackage{authblk}
\usepackage{amsmath,amssymb,amsthm}
\usepackage[numbers]{natbib}
\usepackage{graphicx}
\usepackage{mathtools}
\usepackage{tabularx}
\def\keywords{\vspace{.5em}
{\textit{Keywords}:\,\relax%
}}

\usepackage[outdir=./]{epstopdf}

\title{Estimation of Markovian-regime-switching models with independent regimes}

\author[1,2]{Nigel Bean}
\author[1,2]{Angus Lewis\thanks{This research was supported by the provision of an Australian Government Research Training Program Scholarship, a University of Adelaide postgraduate research scholarship, and an Australian Research Council Discover Project, Grant/Award Number: ARC DP180103106. \\ \texttt{email: angus.lewis@adelaide.edu.au}}}
\author[1,2]{Giang T. Nguyen}

\affil[1]{School of Mathematical Sciences, The University of Adelaide, Australia}
\affil[2]{Australian Research Council Centre of Excellence in Mathematical and Statistical Frontiers (ACEMS) affiliated authors.}

\newtheorem{theorem}{Theorem}
\newtheorem{lemma}[theorem]{Lemma}

\newtheorem{corollary}[theorem]{Corollary}
\newtheorem{example}[theorem]{Example}

\DeclareMathOperator*{\argmax}{arg\,max}

\begin{document}

\maketitle

\begin{abstract}
Markovian-regime-switching (MRS) models are commonly used for modelling economic time series, including electricity prices where \emph{independent regime} models are used, since they can more accurately and succinctly capture electricity price dynamics than \emph{dependent regime} MRS models can. We can think of these independent regime MRS models for electricity prices as a collection of independent AR(1) processes, of which only one process is observed at each time; which is observed is determined by a (hidden) Markov chain. Here we develop novel, computationally feasible methods for MRS models with independent regimes including forward, backward and EM algorithms. The key idea is to augment the hidden process with a counter which records the time since the hidden Markov chain last visited each state that corresponding to an AR(1) process. 
\end{abstract}
\keywords{Electricity price model, forward-backward algorithm, hidden Markov model, Markov-switching time series}

\section{Introduction} 
A commonly used model for economic time series is the Markovian-regime-switching (MRS) model whereby multiple stochastic processes are interweaved by a Markov chain. The general idea is that there exist multiple regimes underlying the observation process, and depending on which regime the system is in, different characteristics are displayed. For example, for stock prices we could suppose that there is a \emph{bull regime} where prices trend upward and are comparatively non-volatile, and a \emph{bear regime} where prices trend downward and are relatively volatile. Our motivating application is electricity prices where it is common to model prices with an MRS model. Due to the fact that electricity cannot currently be stored efficiently, electricity prices show characteristics not seen in typical commodity markets, for example mean reversion, prices spikes, drops and negative prices. MRS models are able to capture these behaviours, and have been popular tools for modelling randomness in electricity markets. Typically MRS models with two or three regimes are used \citep{ethier1998,deng2000,huisman2003,huisman2003regime}: a \emph{base regime} where prices are relatively non-volatile, a \emph{spike regime} where prices are volatile and high, and sometimes a \emph{drop regime} is included, where prices are volatile and low. More broadly, models with Markovian switching find application in biology \citep{albert1991}, weather modelling \citep{thyer2000,zucchini1991}, speech recognition \citep{levinson1983} and more.  

The earliest applications of MRS models for electricity prices are \cite{ethier1998} and \cite{deng2000}. %
Their models specify that prices decay back to base levels following a spike according to an autoregressive process of order 1 (AR(1)). However, this is not consistent with observations from the market where a more immediate return to base levels is observed \citep{janczura2010,escribano2002}. For this reason, \cite{huisman2003regime} introduce a three-regime MRS model which separates the behaviour of base and spike prices. They use one regime to capture base prices, one regime to capture spikes, and one regime to return prices to base levels following a spike. Motivated by the need to capture the distinct and abrupt price spikes in electricity markets, it has become popular to specify MRS models with \emph{independent} regimes, which are able to capture this behaviour without the addition of the extra regime required by \cite{huisman2003regime}. We may think of a dependent regime MRS models as a single process, where the dynamics of the process is governed by the hidden Markov chain, whereas we may think of an independent regime MRS model as a collection of independent AR(1) processes, and at each time \(t\), the hidden Markov chain chooses which process is observed. 

We say a model has independent regimes if, given the hidden regime sequence, the observations generated from each regime are independent of observations generated from any other regime; and we say a model has dependent regimes otherwise.  Independent regime MRS models were introduced in \cite{huisman2003} and have since been popular \citep{weron2004,janczura2012,weron2014}. In these models, typically at least one regime is specified as an AR(1) process. AR(1) processes have a dependence structure between values at successive times which, coupled with the independent regimes assumption in an MRS model, complicates the dependence structure between an observation at time \(t\) and all prior observations -- the dependence between prices is governed by the hidden regime process and is therefore random. It is for this reason that the \emph{forward}, \emph{backward} and \emph{expectation-maximisation} (EM) algorithms for traditional (dependent regime) MRS models or hidden Markov models do not apply. For each time \(t\), the forward algorithm evaluates the probabilities that the hidden process is in each regime given the observed values \emph{up to time} \(t\). The backward algorithm then uses the output of the forward algorithm to calculate the probabilities that the hidden process is in each regime given \emph{all} observations. The EM algorithm is an iterative optimisation algorithm, iterating between an E-step and an M-step, used to find the \emph{maximum likelihood estimates} of parameters  for models with missing data or latent variables -- the E-step is computed by the backward algorithm. 

For simplicity, in this work, we focus on independent regime MRS models with AR(1) and i.i.d.~regimes only, since these are the type of models used in the electricity price modelling literature. However, we believe the methods developed here are more general and apply to Markov-switching processes with regimes that are discrete-time Markov chains generally, and can also be extended to more general autoregressive processes. The most popular method of inference for the models used for electricity pricing is an approximation to the EM algorithm introduced by \cite{janczura2012}, which we show can be unreliable (see \cite{lewis2018} also). Here, a novel, computationally feasible, and exact likelihood-based framework to solve this problem is developed. This work is related to the forward, backward, and EM algorithms for traditional MRS models, and, more closely, to the same algorithms for hidden semi-Markov models, where the idea of augmenting the hidden process with a counter is also used 
\cite{yu2016}. The novel algorithms have complexity \(\mathcal O\left(M^2T^{k+1}k^k\right)\) where \(M<\infty\) is the total number of regimes in the model, \(T\) is the length of the observed data set, and \(k\) is the number of AR(1) regimes in the model. As \(\mathcal O\left(M^2T^{k+1}k^k\right)\) may be impractically large depending on the values of \(T\) and \(k\), we also present an approximation to our algorithm that is \(\mathcal O\left(M^2TD^{k}k^k\right)\) where \(D>0\) is the maximum length of the memory of the AR(1) processes, that is, we specify that the observations \(x_t\) may depend only on \(x_{t-\ell}\) for \(\ell\in\{1,2,...,D\}\), and is independent of any \(x_{t-\ell}\) for \(\ell>D\).

This paper is structured as follows. We formally introduce the MRS model, in particular the independent-regime type in Section~\ref{model intro}, and discuss the approximate parameter inference algorithm of \cite{janczura2012} in Section~\ref{sec:existingmethods}. A novel forward algorithm, which is used to evaluate the likelihood and \emph{filtered state probabilities} for these models, is presented in Section~\ref{myfwd}. Using the outputs of the forward algorithm we develop a novel backward algorithm in Section~\ref{mybkwd}. The backward algorithm is used to evaluate the \emph{smoothed state probabilities}, which are applied in Section~\ref{myEM} to construct an EM algorithm. We introduce the truncated approximations of our algorithms at the end of Section~\ref{myEM}. Section~\ref{simulations} provides simulation evidence that our EM algorithm is consistent and that the truncation approximations are reasonable, while in Section~\ref{application} we apply our algorithms to estimate MRS models for the South Australian wholesale electricity market. Finally, we make concluding remarks in Section~\ref{conclusion}.

\subsection{MRS models: A brief introduction}\label{model intro}
An MRS model is built from two pieces, an unobservable regime sequence, \(\{R_t\}_{t\in \mathbb N}\), which is a finite-state Markov Chain, and an observation sequence, \(\{X_t\}_{t\in \mathbb N}\). Let us denote the state space of the hidden regime sequence as \(\mathcal S = \{1,2,\dots,M<\infty\}\), and the transition matrix as \(P=[p_{ij}]_{i,j \in \mathcal S}\). 

The simplest MRS model is the hidden Markov model (HMM) where observations \(X_t\) take values in a discrete set, and \(X_t\) is independent of \(X_{t-1},\dots,X_0\) and \(X_{t+1},X_{t+2},\dots\) given the regime at time \(t\), \(R_t\). In general, MRS models are specified in terms of distributions that allow dependence on past observations, given the current regime. That is, the model defines distributions,  
\[X_t|\{R_t,X_{t-1},X_{t-2},\dots,X_0\}\sim F^{R_t},\]
for some distribution \(F^{R_t}\). The MRS model, as introduced by \cite{hamilton1989,hamilton1990}, specifies that 
\(X_t|\{R_t,X_{t-1},X_{t-2},\dots,X_0\}\) 
follows some time-series model (an autoregressive process of order \(p\), for example) with dependence on a finite number of past observations, but \emph{not on \(R_0,..,R_{t-1}\)}. That is, the dependence structure does not take into account which regime the past observations belong to. For example, the following is a dependent regime MRS model.
\begin{example}[An MRS model with dependent regimes]\label{ex: Dependent MRS model 3 regime sim}
    Let \(\mathcal S = \{1,2\}\), and \(p_{11}=p_{22}=0.9\),
    and specify 
    \begin{align*}
        &X_t | \{R_t=1, X_{t-1}, X_{t-2}, \dots,X_0\} = 0.6 X_{t-1}+\varepsilon_t^{\left(1\right)},\\
        &X_t | \{R_t=2, X_{t-1}, X_{t-2}, \dots,X_0\} = 1+0.9 X_{t-1} + \varepsilon_t^{\left(2\right)}.
    \end{align*}
    for \(\varepsilon_t^{\left(i\right)}\sim\) i.i.d.~N(0,1) for \(i=1,2\). So \(X_t\) follows AR(1) dynamics in both Regimes 1 and 2. This is a dependent-regime MRS model since \(X_t\) depends on \(X_{t-1}\) regardless of which regime the lagged observation, \(X_{t-1}\), came from. Figure~\ref{fig:simulated examples} (Left) shows a simulation of this model. 
\end{example}
\begin{figure}
\begin{center}
\includegraphics[width=0.5\textwidth]{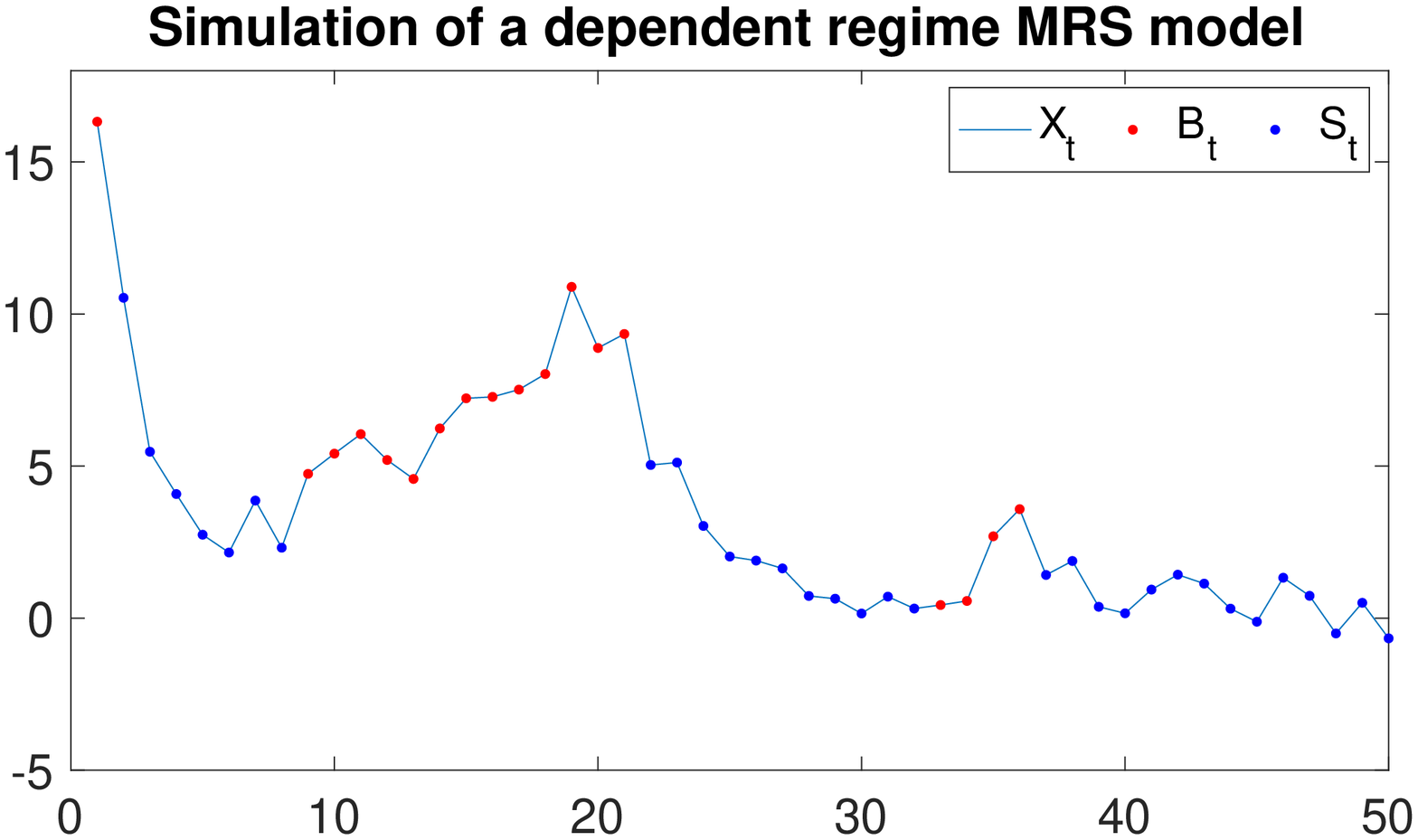}%
\includegraphics[width=0.5\textwidth]{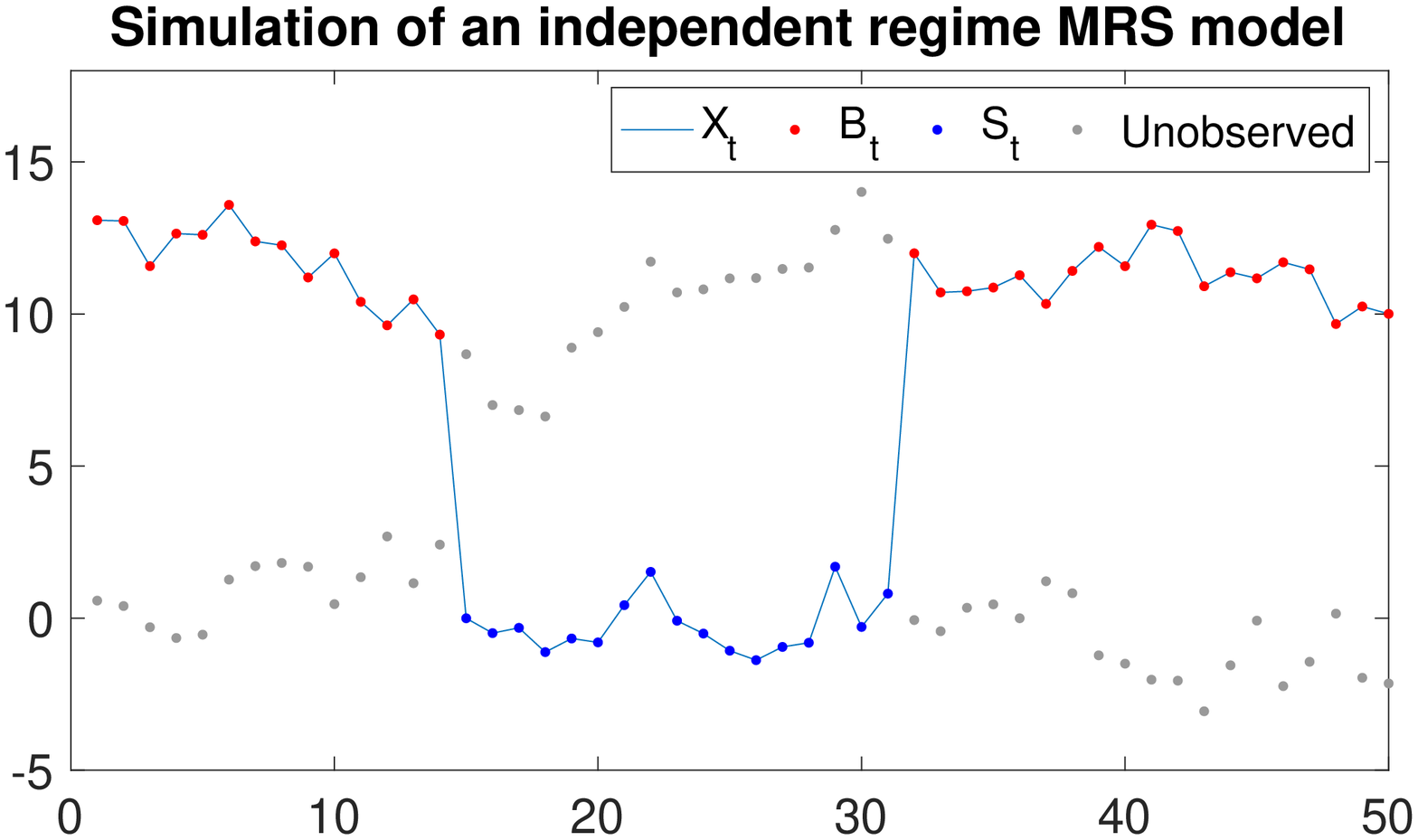}
\caption{(Left) Simulation of the dependent regime MRS model example. (Right) Simulation of the independent regime MRS model example. Notice that after a change of regime the independent regime model shows a distinct change in characteristics. }
\label{fig:simulated examples}
\end{center}
\end{figure}
In this paper we relax the assumption that \emph{the current observation \(X_t\) is conditionally independent of \(R_0,..,R_{t-1}\)}, in order to increase flexibility in these models. In particular, we consider models where, given \(R_t=i\), \(X_t\) depends only on lagged values from Regime \(i\), thus the dependence structure is random as it is a function of \(R_0,\dots,R_{t-1}\). The following is an example of an independent regime MRS model.
\begin{example}[An MRS model with independent regimes]\label{ex: Independent MRS model 2 regime sim}
    Let \(\mathcal S = \{1,2\}\), and \(p_{11}=p_{22}=0.9\), and define the following AR(1) processes 
    \begin{align*}
        B_t &=  0.6 B_{t-1} + \varepsilon_t^B,\\
	S_t &= 1 + 0.9S_{t-1} + \varepsilon_t^S,
    \end{align*}
    where \(\varepsilon_t^B\) and \(\varepsilon_t^S\) are sequences of i.i.d.~N(0,1) random variables. Then, construct the MRS model as follows
    \begin{equation*}
        X_t = \begin{cases} B_t, & \text{ if } R_t = 1, \\ S_t ,& \text{ if } R_t = 2. \end{cases}
    \end{equation*}
    Figure~\ref{fig:simulated examples} (Right) shows a simulation of this model. 
\end{example}

A precise definition of dependent and independent regime models is the following. Define the sets \(\mathcal A_i := \{t\in \mathbb N : R_t = i\}\), \(i \in \mathcal S\). We say that a model has \emph{independent regimes} if, given the regime sequence, the sets \(\{X_{t} : t \in \mathcal A_i\}\), \(i \in \mathcal S\), are independent. Otherwise, it is a \emph{dependent regime} model.

\subsection{Existing methods}\label{sec:existingmethods}
For the simplest form of MRS model, the HMM, likelihood evaluation and maximisation algorithms were first developed in a series of papers, \cite{baum1966}, \cite{baum1967}, and \cite{baum1970}, and subsequent work on MRS models is typically closely related to this. The first algorithms for the more general \emph{dependent regime} MRS model were presented by \cite{hamilton1989,hamilton1990}. The main issue for maximum likelihood estimation of models with hidden regimes is that the regime sequence is unobserved, thus to naively evaluate the likelihood requires calculation of the marginal distribution 
\begin{equation}\label{likelihoodeq}
L\left(\boldsymbol{\theta}\right) := f^{\boldsymbol \theta}_{\boldsymbol X}\left(\boldsymbol x\right) = \sum_{\boldsymbol R \in \mathcal S^{T+1}} f^{\boldsymbol \theta}_{\boldsymbol X,\boldsymbol R}\left(\boldsymbol x , \boldsymbol R \right) = \sum_{\boldsymbol R \in \mathcal S^{T+1}} f^{\boldsymbol \theta}_{\boldsymbol X|\boldsymbol R}\left(\boldsymbol x |\boldsymbol R\right)f^{\boldsymbol \theta}_{\boldsymbol R}\left(\boldsymbol R\right),
\end{equation} 
where \(f_{\boldsymbol Y}^{\boldsymbol \theta}\left(\cdot\right)\) denotes the distribution function of a random vector \(\boldsymbol Y\) with parameters \(\boldsymbol \theta\), \(\boldsymbol x=\left(x_0,\dots,x_T\right)\) is a sequence of observed values, and \(\mathcal S^{T+1}\) is the space of all possible regime sequences of length \(T+1\), \(\boldsymbol R=\left(R_0,...,R_T\right)\). The number of sequences in \(\mathcal S^{T+1} \) is \(M^{T+1}\) which, for most realistic datasets, is computationally infeasible to evaluate in this form. In the context of HMMs, the sum (\ref{likelihoodeq}) is made computationally feasible by the forward algorithm \citep{baum1970}, and the maximisation of the likelihood is commonly performed via the Baum-Welch algorithm \citep{baum1970}, which is a specific case of the EM algorithm \citep{dempster1977} and uses the backward algorithm \citep{baum1970}. 

The works of \cite{hamilton1989,hamilton1990} extend the methods for HMMs to MRS models with dependent regimes by adapting the forward algorithm, developing a new algorithm to replace the backward algorithm and constructing an EM algorithm. \cite{kim1994} refines the work of Hamilton, developing a more efficient implementation of Hamilton's smoothing algorithm. Kim's algorithm is similar to the backward algorithm for HMMs. 

Relevant to this paper, \cite{janczura2012} extend Hamilton's work and develop an \emph{approximate} algorithm for MRS models with \emph{independent} regimes, which we label the \emph{EM-like algorithm} since it resembles Hamilton's EM algorithm. However, it is not an example of the EM algorithm and so none of the EM theory holds. In Sections~\ref{sec: hamilton fwd}--\ref{sec:EM-like}, we briefly review the work of \cite{hamilton1990}, \cite{kim1994} and \cite{janczura2012} in order to provide the motivation and background for our work.

\subsubsection{Likelihood evaluation for dependent-regime models: The forward algorithm}\label{sec: hamilton fwd}

Define \(\boldsymbol{x}_{r:s} = \left(x_r,x_{r+1},\dots,x_s\right)\) for \(r \leq s\) and write the likelihood as 
\(L\left(\boldsymbol{\theta}\right) = f^{\boldsymbol \theta}_{X_0}\left(x_0\right) \prod\limits_{t=1}^T  f^{\boldsymbol \theta}_{X_t|\boldsymbol X_{0:t-1}}\left(x_t|\boldsymbol x_{0:t-1}\right).\) 
The forward algorithm \citep{hamilton1990} calculates \(f^{\boldsymbol \theta}_{X_0}\) and \(f^{\boldsymbol \theta}_{X_t|\boldsymbol X_{0:t-1}}\) for \(t=1,2,\dots,T,\) from which it is straightforward to calculate the likelihood or loglikelihood:
\paragraph{Algorithm 1: The forward algorithm \citep{hamilton1990}}
\begin{enumerate}
\item[Step 1.] Initialise the algorithm with values \(\mathbb P^{\boldsymbol\theta}\left(R_0=i\right)=\pi_i\), which may be assumed to be known \emph{a priori}, or left as parameters to be inferred.
\item[Step 2.]The term, \(f^{\boldsymbol \theta}_{X_0}\left(x_0\right)\), is calculated as 
\(f^{\boldsymbol \theta}_{X_0}\left(x_0\right) = \sum\limits_{i \in \mathcal S} f^{\boldsymbol \theta}_{X_0|R_0}\left(x_0| i\right)\mathbb  P^{\boldsymbol \theta}\left(R_0=i\right),\) where the density \(f^{\boldsymbol \theta}_{X_0|R_0}\) is known from the model specification. 
\item[Step 3.] For \(t =1 ,\dots, T\):
\begin{align*}&f^{\boldsymbol \theta}_{X_t|\boldsymbol X_{0:t-1}}\left(x_t|\boldsymbol x_{0:t-1}\right) \\&\quad=  \sum\limits_{i \in \mathcal S} f^{\boldsymbol \theta}_{X_t|R_t,\boldsymbol X_{0:t-1}}\left(x_t| i,\boldsymbol x_{0:t-1}\right) \sum\limits_{j \in \mathcal S}  \mathbb P^{\boldsymbol \theta}\left(R_{t-1} = j| \boldsymbol x_{0:t-1}\right)p_{ji},\end{align*}
where \(f^{\boldsymbol \theta}_{X_t|R_t,\boldsymbol X_{0:t-1}}\) is also known from the model specification. Furthermore, the probabilities \(\mathbb P^{\boldsymbol \theta}\left(R_{t-1} = j| \boldsymbol x_{0:t-1}\right)\) can be calculated using Bayes' Theorem, 
\begin{align}
&\mathbb P^{\boldsymbol \theta}\left(R_{t-1} = j| \boldsymbol x_{0:t-1}\right)\nonumber \\&= \cfrac{f^{\boldsymbol \theta}_{X_{t-1}|R_{t-1},\boldsymbol X_{0:t-2}}\left(x_{t-1}| j, \boldsymbol x_{0:t-2}\right)\sum\limits_{i\in \mathcal S} p_{ij}\mathbb P^{\boldsymbol \theta}\left( R_{t-2} = i|\boldsymbol x_{0:t-2}\right)}{f^{\boldsymbol \theta}_{X_{t-1}|\boldsymbol X_{0:t-2}}\left(x_{t-1}|\boldsymbol x_{0:t-2}\right)},\label{fwdprobs}
\end{align}
for \(t\geq 2\). These are known as the \emph{forward/filtered probabilities}. 
\end{enumerate}

The quantities \(\mathbb P^{\boldsymbol \theta}\left(R_t=i|\boldsymbol x_{0:t-1}\right)=\sum\limits_{j\in \mathcal S} p_{ji}\mathbb P^{\boldsymbol \theta}\left( R_{t-1} = j|\boldsymbol x_{0:t-1}\right)\) are known as the \emph{prediction probabilities}. 
In some applications, the forward and prediction probabilities may be quantities of interest in their own right, and they also appear as inputs to the backward algorithm (see Section \ref{sec: hamilton kim}).



\subsubsection{Maximum likelihood for dependent-regime models: The EM algorithm}\label{sec: hamilton kim}
Hamilton's forward algorithm \citep{hamilton1990} is a computationally feasible way to evaluate the loglikelihood, from which it is possible to use black-box optimisation methods to find the MLEs. However, it is common to use the EM algorithm \citep{dempster1977} instead, particularly when the E-step and M-step of the algorithm are available in closed form.  The EM algorithm proceeds by iterating between the E-step, constructing the function \(Q\left(\boldsymbol \theta , \boldsymbol \theta_n\right) = \mathbb E\left[ \log f^{\boldsymbol \theta}_{\boldsymbol X, \boldsymbol R}\left(\boldsymbol{x},\boldsymbol R\right)\mid\boldsymbol{x};\boldsymbol{\theta}_n\right]\), and the M-step, maximising \(Q\) with respect to \(\boldsymbol \theta \in  \Theta\), where \(\Theta\) is the parameter space. This results in a sequence \(\{\boldsymbol\theta_n\}_{n \in \mathbb N}\) that converges to a local maximiser of the loglikelihood. 

The EM algorithm for MRS models with \emph{dependent} regimes proceeds as follows \citep{hamilton1990}. Define the random variable \(\eta_{ij}\) as the number of transitions from state \(i\) to state \(j\) in the sequence \(\boldsymbol R = \left(R_0,\dots,R_T\right)\) and let \(\mathbb I\left(\cdot\right)\) be the indicator function. The joint log-density of \(\boldsymbol{x}\) and \(\boldsymbol{R}\) can be written as 
\begin{align*}
    &\log f^{\boldsymbol \theta}_{\boldsymbol X, \boldsymbol R}\left(\boldsymbol{x},\boldsymbol R\right) 
    \\&= \sum_{j \in \mathcal S}{\mathbb I\left(R_0 = j\right)} \log f^{\boldsymbol \theta}_{X_0|R_0}\left(x_0|j\right) \\&\quad{}+ \sum_{t=1}^T \sum_{j \in \mathcal S} {\mathbb I\left(R_t = j\right)}\log  f^{\boldsymbol \theta}_{X_t|R_t,\boldsymbol X_{0:t-1}}\left(x_t| j, \boldsymbol x_{0:t-1}\right)\\&\quad{}+ \sum_{i,j \in \mathcal S}{\eta_{ij}}\log p_{ij} + \sum_{j\in \mathcal S}\mathbb I\left(R_0=j\right) \log \pi_j,
\end{align*}
where \(\pi_j\) denotes \(\mathbb P^{\theta}\left(R_0=i\right)\). In the \(n\)th iteration, \(n\geq 0\), for the E-step, taking the conditional expectation given parameters \(\boldsymbol{\theta}_n\) and observed values \(\boldsymbol{x}_{0:T}\) yields
\begin{align*}
    Q\left(\boldsymbol{\theta},\boldsymbol{\theta}_n\right) &= \sum_{j \in \mathcal S}{\mathbb P^{\boldsymbol \theta_n}\left(R_0 = j|\boldsymbol x_{0:T}\right)} \log f^{\boldsymbol \theta}_{X_0|R_0}\left(x_0| j\right)\\&\quad{} + \sum_{t=1}^T \sum_{j \in \mathcal S} {\mathbb P^{\boldsymbol \theta_n}\left(R_t = j|\boldsymbol x_{0:T}\right)}\log  f^{\boldsymbol \theta}_{X_t|R_t,\boldsymbol X_{0:t-1}}\left(x_t| j,\boldsymbol x_{0:t-1}\right)\\&\qquad + \sum_{i,j \in \mathcal S}{\mathbb E [\eta_{ij}|\boldsymbol{x}_{0:T};{\boldsymbol \theta_n}]}\log p_{ij} + \sum_{j\in \mathcal S}\mathbb P^{\boldsymbol \theta_n}\left(R_0=j|\boldsymbol x_{0:T}\right) \log \pi_j,
\end{align*}
where the expectation \(\mathbb E[\eta_{ij}|\boldsymbol x_{0:T}; \boldsymbol \theta_n]=\sum\limits_{t=1}^T\mathbb P^{\boldsymbol \theta_n}\left(R_t=j,R_{t-1}=i|\boldsymbol x_{0:T}\right)\). The densities \(f^{\boldsymbol \theta}_{X_0|R_0}\) and \(f^{\boldsymbol \theta}_{X_t|R_t,\boldsymbol X_{0:t-1}}\), are given by the model specification. 

The \emph{smoothed probabilities}, \(\mathbb P^{\boldsymbol \theta_n}\left(R_t = j|\boldsymbol x_{0:T}\right)\) and \(\mathbb P^{\boldsymbol \theta_n}\left(R_t = j, R_{t-1} = i \mid \boldsymbol x_{0:T}\right)\), required to construct \(Q\) are obtained using a backward recursion after running the forward algorithm with parameters \(\boldsymbol{\theta}_n\), and storing the forward and prediction probabilities. Developed by \cite{kim1994}, this backward recursion is in Algorithm 2, below.
\paragraph{Algorithm 2: The backward algorithm \citep{kim1994}} 
\begin{enumerate} 
\item[Step 1.] Evaluate \(\mathbb P^{\boldsymbol \theta_n}\left(R_T = j|\boldsymbol x_{0:T}\right)\) using the forward algorithm (Algorithm 1). 
\item[Step 2.] For \(t = T-1,\dots,0,\)
\begin{align*}
\mathbb P^{\boldsymbol \theta_n}\left(R_t = j, R_{t+1} = i | \boldsymbol x_{0:T}\right) 
    &= p_{ji}^{\left(n\right)}\cfrac{\mathbb P^{\boldsymbol \theta_n}\left(R_t = j|\boldsymbol x_{0:t}\right)  \mathbb P^{\boldsymbol \theta_n}\left(R_{t+1} = i|\boldsymbol x_{0:T}\right)}{\mathbb P^{\boldsymbol \theta_n}\left(R_{t+1} = i|\boldsymbol x_{0:t}\right)},
\\\mathbb P^{\boldsymbol \theta_n}\left(R_t = j|\boldsymbol x_{0:T}\right) 
    &= \sum_{i \in \mathcal S}\mathbb P^{\boldsymbol \theta_n}\left(R_t = j, R_{t+1} = i | \boldsymbol x_{0:T}\right),
\end{align*}
where \(p_{ij}^{\left(n\right)}\) means the \(p_{ij}\) parameter under \(\boldsymbol{\theta}_n\). 
\end{enumerate}

After executing Kim's backward algorithm, we can construct the function \(Q\). In the M-step, the maximisers of \(Q\left(\cdot,\boldsymbol{\theta}_{n}\right)\) are found. In the \emph{dependent regime} model, if the process is in Regime \(j\) at time \(t\), then the observations evolve according to \(X_t=\alpha_j + \phi_j X_{t-1} + \sigma_j \varepsilon_t\) where \(\alpha_j,\) \(\phi_j\) and \(\sigma_j\) are parameters, and \(\{\varepsilon_t\}\sim\mbox{N}\left(0,1\right)\). Recall that, for the dependent regime model, \(X_{t-1}\) is the last observed value (regardless of which regime generated it). The maximiser of \(Q\) at the \(\left(n+1\right)\)th iteration of the EM algorithm, \(\boldsymbol \theta ^{n+1}\), for \( n\geq 0\), is the following system of equations \citep{hamilton1990,janczura2012}:
\begin{align*}
\phi^{\left(n+1\right)}_{j} &= \cfrac{\sum\limits_{t=1}^T\mathbb P^{\boldsymbol \theta_n}\left(R_t = j|\boldsymbol x_{0:T}\right)x_{t-1}B_{1,t}^{\left(i\right)}}{\sum\limits_{t=1}^T\mathbb P^{\boldsymbol \theta_n}\left(R_t = j|\boldsymbol x_{0:T}\right)x_{t-1}B_{2,t}^{\left(i\right)}},
\\
\alpha^{\left(n+1\right)}_{j}& = \cfrac{\sum\limits_{t=1}^T\mathbb P^{\boldsymbol \theta_n}\left(R_t = j|\boldsymbol x_{0:T}\right)\left(x_t-\phi^{\left(n+1\right)}_{j}x_{t-1}\right)}{\sum\limits_{t=1}^T\mathbb P^{\boldsymbol \theta_n}\left(R_t = j|\boldsymbol x_{0:T}\right)} ,
\\
\left(\sigma^2_j\right)^{\left(n+1\right)} &= \cfrac{\sum\limits_{t=1}^T\mathbb P^{\boldsymbol \theta_n}\left(R_t = j|\boldsymbol x_{0:T}\right)\left(x_t-\alpha^{\left(n+1\right)}_{j}-\phi^{\left(n+1\right)}_{j} x_{t-1}\right)^2}{\sum\limits_{t=1}^T\mathbb P^{\boldsymbol \theta_n}\left(R_t = j|\boldsymbol x_{0:T}\right)}, \qquad 
\\ 
\mbox{where }&
\\B_{1,t}^{\left(i\right)} &= x_t - x_{t-1}-\cfrac{\sum\limits_{s=1}^T\mathbb P^{\boldsymbol \theta_n}\left(R_s = j|\boldsymbol x_{0:T}\right)\left(x_s-x_{s-1}\right)}{\sum\limits_{s=1}^T\mathbb P^{\boldsymbol \theta_n}\left(R_s = j|\boldsymbol x_{0:T}\right)},\quad
\\\mbox{ and } \quad
B_{2,t}^{\left(i\right)} &= \cfrac{\sum\limits_{s=1}^T\mathbb P^{\boldsymbol \theta_n}\left(R_s = j|\boldsymbol x_{0:T}\right)x_{s-1}}{\sum\limits_{s=1}^T\mathbb P^{\boldsymbol \theta_n}\left(R_s = j|\boldsymbol x_{0:T}\right)}-x_{t-1}.
\end{align*}
In general, the switching probabilities are updated using the following \citep{kim1994}
\begin{align}\label{eqn: simple EM updates pij}
    p_{ij}^{\left(n+1\right)} = \cfrac{\sum\limits_{t=1}^T\mathbb P^{\boldsymbol \theta_n}\left(R_t = j|\boldsymbol x_{0:T}\right)\cfrac{p_{ij}^{\left(n\right)}\mathbb P^{\boldsymbol \theta_n}\left(R_{t-1} = i|\boldsymbol x_{0:t-1}\right)}{\mathbb P^{\boldsymbol \theta_n}\left(R_t = j|\boldsymbol x_{0:t-1}\right)}}{\sum\limits_{t=1}^T\mathbb P^{\boldsymbol \theta_n}\left(R_{t-1} = i|\boldsymbol x_{0:T}\right)},
\end{align}
which rely on the smoothed, forward and prediction probabilities. For i.i.d.~regimes the M-step can often be derived analytically; however, such expressions are not required for our discussion since there is no dependence on lagged values in these regimes and so they are omitted. 

Thus, we implement the EM algorithm by initialising it with a guess of the true parameters, then alternating between the forward and backward algorithms (the E-step) and calculating the maximisers of \(Q\) (the M-step). The algorithm terminates when the step size is below a prespecified tolerance, i.e.~\linebreak\mbox{\(|\boldsymbol{\theta}_{n+1} - \boldsymbol{\theta}_{n}|_\infty < e\)} where \(e\) is some small tolerance.

\subsubsection{Approximate maximum likelihood for independent-regime models: The EM-like algorithm}\label{sec:EM-like}
For independent-regime MRS models the EM algorithm is computationally infeasible if the densities \(f^{\boldsymbol \theta}_{X_t|R_t,\boldsymbol X_{0:t-1}}\), \(R_t \in \mathcal S\), \(t = 1,\dots,T\), are computed naively, as this is a \(\mathcal O\left(M^t\right)\) calculation for each \(t=1,...,T\), where \(M<\infty\) is the number of regimes in the model. On the other hand, if we use our proposed foward algorithm, introduced in Section \ref{myfwd}), calculation of these densities is \(\mathcal O\left(t^{k+1}\right)\) for each \(t=1,...,T\), which may be computationally feasible when \(T\) and \(k\) are not too large.

Developed by \cite{janczura2012}, the EM-like algorithm is an approximation to the EM algorithm. For independent regime MRS models, the EM-like algorithm overcomes the problem of computational infeasibility by replacing lagged values for Regime \(i\) (assuming this is an AR(1) regime) with approximations, \(\Tilde{b}_{t-1,i}^{\left(n\right)}\). These are described as the expectations \citep{janczura2012}, 
\(
    \mathbb E[B^i_t|\boldsymbol x_{0:t};\boldsymbol{\theta}_n],
\)
where \[B_t^i := \mathbb I \left(R_t = i\right)x_t + \mathbb I \left(R_t \neq i\right) \left(\alpha_i+\phi_i B_{t-1}^i+\sigma_i \varepsilon_t^i\right).\] Simply put, wherever \(x_{t-1}\) appears verbatim in an expression related to Regime \(i\) in the EM algorithm in Sections \ref{sec: hamilton fwd}-\ref{sec: hamilton kim}, it is replaced with \(\Tilde{b}_{t-1,i}^{\left(n\right)}\) at the \(n^{\text{th}}\) iteration. 
The \(\Tilde{b}_{t,i}^{\left(n\right)}\) are calculated recursively as 
\begin{align}\label{eqn: b_t-1^i}
    \Tilde{b}_{t,i}^{\left(n\right)} 
    &= \Tilde{\mathbb P}^{\boldsymbol{\theta}_n}\left(R_t=i|\boldsymbol x_{0:t}\right)x_t + \Tilde{\mathbb P}^{ \boldsymbol{\theta}_n}\left(R_t \neq i|\boldsymbol x_{0:t-1}\right)\left(\alpha_i + \phi_i \Tilde{b}_{t-1,i}^{\left(n\right)} \right),
\end{align}
where \(\Tilde{\mathbb P}^{\boldsymbol{\theta}_n}\left(R_t=i|\boldsymbol x_{0:t}\right)\) and \(\Tilde{\mathbb P}^{\boldsymbol{\theta}_n}\left(R_t\neq i|\boldsymbol x_{0:t}\right)\) are given by the forward algorithm which is part of the EM-like procedure.
Janczura and Weron, \cite{janczura2012}, conduct simulation studies and show that this algorithm seems to work well for the datasets they generate. However, no theoretical results are available that show convergence of, or error bounds for, the EM-like algorithm; in particular, there is no guarantee that the parameter estimates produced by the EM-like algorithm are consistent. In contrasts, our algorithms rest on the theory of the EM algorithm.

We can construct examples of independent-regime MRS models where the EM-like algorithm fails to get close to the true parameter values. 
\begin{example}\label{emlikefaileg}
Consider the following independent-regime MRS model,
\begin{align}\label{ex:1}
    X_t = \begin{cases}
            B_t, & \text{ if } R_t = 1,
            \\  Y_t, & \text{ if } R_t = 2, 
        \end{cases} 
\end{align}
where \(B_t\) is an AR(1) process, \(B_t = 0.95 B_{t-1} + \sqrt{0.2} \varepsilon_t,\) 
with \(\{\varepsilon_t\}\) being a sequence of i.i.d.~N\(\left(0,1\right)\) random variables, \(Y_t\) is an i.i.d.~sequence of N\(\left(2,1\right)\) random variables, and \(\{R_t\}_{t \in \mathbb N}\) is a Markov chain with state space \(\mathcal S = \{1,2\}\), transition matrix entries \(p_{11}=0.5\) and \(p_{22}=0.8\), and initial probability distribution \(\left(1,0\right)\), so the process always starts in Regime 1.
We simulated 20 realisations of length \(T=2000\) from this model and used the EM-like algorithm to try to recover the true parameters. 
To give the algorithm the best chance of converging to the true parameters, we initialise the EM-like algorithm at the true parameter values. The parameters recovered by the EM-like algorithm are summarised in Figure~\ref{emlikeexamplefig}. %
\begin{figure}
    \centering
    \includegraphics[width = \textwidth, trim = {75 0 75 0}, clip]{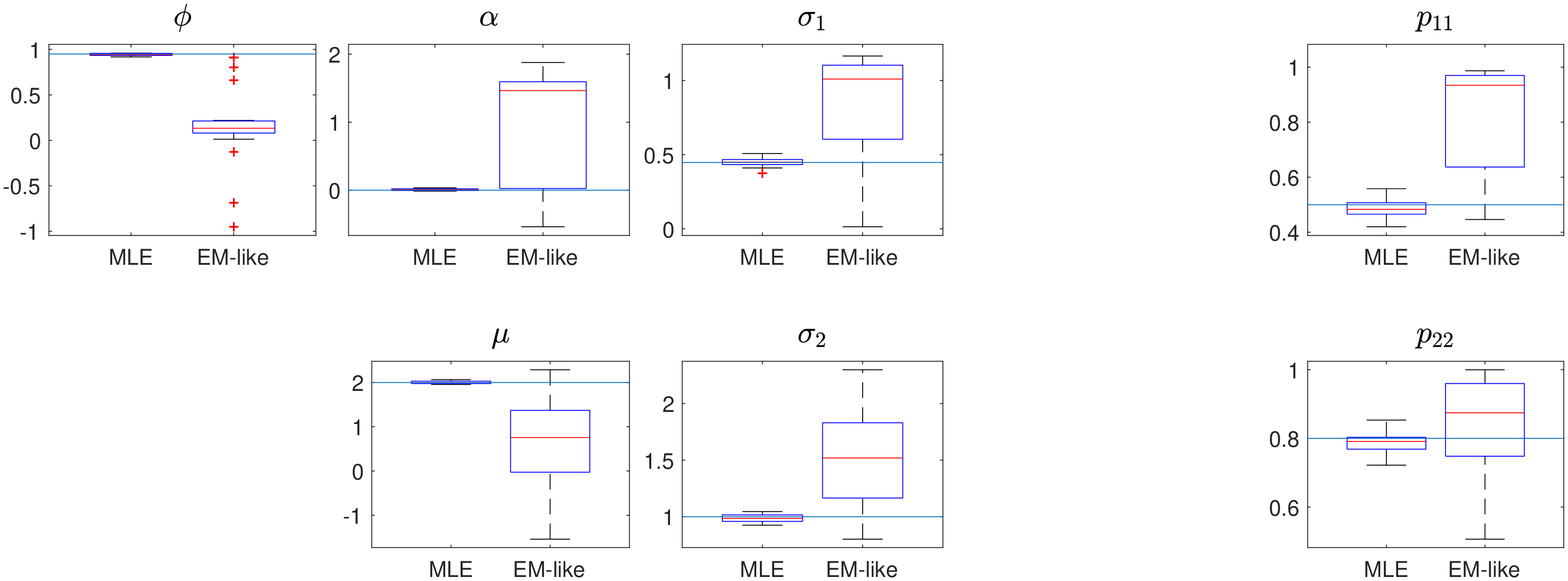}
    \caption{Boxplots of the parameters recovered by the EM-like algorithm (Right) and the MLEs recovered by our EM algorithm (Left) for Example \ref{emlikefaileg}. The blue line represents the true parameter value. Notice that the EM-like algorithm is not able to recover the parameters, while the EM algorithm performs relatively well.}
    \label{emlikeexamplefig}
\end{figure} %
For comparison, the MLEs obtained using our EM algorithm (Section \ref{myEM}) are also shown. Notice in Figure~\ref{emlikeexamplefig} that the EM-like algorithm performs poorly, while our exact method performs much better.
\end{example}

\section{A novel forward algorithm}\label{myfwd}
The general idea of the novel algorithms presented in this paper is to augment the hidden Markov chain with counters that record the last time each AR(1) regime was visited. This augmented process is a Markov chain, and similar arguments to those used to construct the forward-backward algorithm for MRS models with dependent regimes can be used to construct a forward and backward algorithms for these models. Our methods are related to the forward and backward algorithms for hidden semi-Markov models, where the hidden process is also augmented with a counter and the augmented hidden process is a Markov chain \citep{yu2016}. Though similar, the algorithms for hidden semi-Markov models (HSMMs) are not applicable to the models considered here since for HSMMs the counter counts the number of time steps since the last change of regime, whereas here the counters count the number of transitions since the last visit to a regime. Furthermore, in HSMMs, the counters do not appear in the conditional densities of the observations, as they do here.    

\subsection{The augmented hidden Markov chain}
For the following, suppose that the first \(k\) states, \(\mathcal S_{AR} = \{1,\dotsc,k<M\}\), correspond to AR(1) processes and all other regimes are i.i.d. That is, we have the following independent regime MRS model 
    \[ 
        X_t = \begin{cases}
        B_{t}^1, & \mbox{if } R_t=1, \\
        \vdots & \\
        B_{t}^k, & \mbox{if } R_t=k, \\
        S_t^{k+1}, & \mbox{if } R_t = k+1, \\
        \vdots & \\ 
        S_t^{M}, & \mbox{if } R_t = M, \\
        \end{cases}
    \]
    where 
    \(B_t^i = \alpha_i + \phi_i B^i_{t-1} + \sigma_i \varepsilon_t^i\)
    are AR(1) and 
    \(S_t^{j}\) are i.i.d. Our arguments also hold for \(k=M\) with only slight modification, but here we treat the case \(k<M\) only, since these are the types of models relevant to our application. 
    
    Now, define another Markov chain 
\[\{\boldsymbol{H}_t\}_{t \in \mathbb N}:=\{\left(\boldsymbol{N}_t,R_t\right)\}_{t\in \mathbb N}=\{\left(N_{t,1},\dotsc,N_{t,k},R_t\right)\}_{t\in \mathbb N},\] where \(N_{t,j}\in \mathbb N_+\) counts the number of time steps since the process \(\{R_t\}\) was last in Regime \(j\) before time \(t\), for each AR(1) regime \(j=1,\dots,k\). When there is no time \(\tau \in \{0,1,\dots,t-1\}\) with \(R_\tau = j\), then we set \(N_{t,j}=t+1\). Thus the augmented Markov chain \(\{\boldsymbol{H}_t\}_{t \in \mathbb N}\) lives on the state space \(\mathbb N_+^k \times \mathcal S.\) 

To describe the transitions of the Markov chain \(\{\boldsymbol{H}_t\}\), let \linebreak{\(\boldsymbol{n} := \left(n_1,\dots,n_k\right)\in~\mathbb N_+ ^k \)} be an arbitrary vector of counters, with \(n_r \neq n_s\) for \(r\neq s\) or, at time \(t\), \(n_r=n_s=t+1\) (it is possible that neither state \(r\) nor state \(s\) have been visited by time \(t\)). Also, define \(\boldsymbol{1}\) to be a row vector of ones of length \(k\), \(\boldsymbol e_i \) to be a row vector of length \(k\) with all entries being 0 except the \(i^{\text{th}}\) entry which is 1, and 
\(\boldsymbol{n}^{\left(-i\right)} := \boldsymbol{n} - n_{i}\boldsymbol{e}_i = \left(n_{1},\dots,n_{i-1},0,n_{i+1},\dots,n_{k}\right).\)

The transition probabilities of \(\{\boldsymbol{H}_t\}\) are 
\begin{align}
    &\mathbb P^{\boldsymbol{\theta}}\left(\boldsymbol{H}_{t+1} = \left(\boldsymbol{n}_{t+1},j\right) | \boldsymbol{H}_{t}= \left(\boldsymbol{n}_{t},i\right)  \right)\nonumber
    \\&\quad{}= 
    \begin{cases} 
        p_{ij} & \text{for } i \in \mathcal S_{AR}^c,  j\in \mathcal S,  \boldsymbol n_{t+1} = \boldsymbol{n}_t+\boldsymbol{1},\\
        p_{ij} & \text{for } i \in \mathcal S_{AR},  j\in \mathcal S,   \boldsymbol n_{t+1} = \boldsymbol{n}^{\left(-i\right)}_t+\boldsymbol{1},\\
        0 & \text{otherwise}.
    \end{cases}\label{eqn: transitions of H_t}
\end{align}
In words, when the current state is \(\boldsymbol{H}_{t}=\left(\boldsymbol{n}_t,i\right)\) and \(i\) is not an AR(1) regime (so there is no counter associated with state \(i\)), then, at time \({t+1}\), \(R_t\) transitions to state \(j\) with probability \(p_{ij}\) and all the counters are advanced by \(1\) to \(\boldsymbol{n}_{t+1} = \boldsymbol{n}_t+\boldsymbol{1}\), since there has been one more time step since \(\{R_{t}\}\) was last in any state with a counter (any state in \(\mathcal S_{AR}\)). When the current state is \(\boldsymbol{H}_{t}=\left(\boldsymbol{n}_t,i\right)\), where \(i\) is an AR(1) regime, then \(R_t\) transitions to any state \(j \in \mathcal S\) with probability \(p_{ij}\), the counter for Regime \(i\), \(n_{{t+1},i}\), is set to 1, since the last time in state \(i\) was \(t\), and all other counters are advanced by \(1\). All other transition probabilities for \(\{\boldsymbol{H}_t\}\) are 0.

The state space of \(\{\boldsymbol{H}_t\}\) is countably infinite. However, due to the way \(\{\boldsymbol{H}_t\}_{t \in \mathbb N}\) is initialised and evolves, many states are inaccessible for \(\{\boldsymbol{H}_t\}\), \linebreak\(t \in \{0,1,\dots,T<\infty\}\), and this makes our algorithm computationally feasible. 
Specifically, we suppose that the Markov chain \(\{\boldsymbol{H}_t\}\) is initialised with the probability distribution
\begin{align}\label{eqn: init distn of H_t}
    \mathbb P\left(\boldsymbol{H}_0 = \left(n_{0,1},\dots,n_{0,k},j\right)\right) = \begin{cases} \pi_j,& \text{for }j \in \mathcal S, \text{ } n_{0,i}=1, \text{ for all } i \in \mathcal S_{AR}, \\
    0, & \text{otherwise}. \end{cases}
\end{align}
The distribution \(\boldsymbol{\pi}:=\left(\pi_1,\dots,\pi_M\right)\) can be any proper probability distribution. However, in line with existing algorithms for dependent regime MRS models, it can be either the stationary distribution of \(\{R_t\}\), or a point mass on a single state, or, when used as part of the EM algorithm, the probabilities \(\mathbb P^{\boldsymbol{\theta}_n}\left(\boldsymbol{H}_0|\boldsymbol{x}_{0:T}\right)\) calculated at the previous iteration of the EM algorithm. The following lemma gives all the states that \(\{\boldsymbol{H}_t\}\) can be in at time \(t>0\).
\begin{lemma}\label{lemma: S^t}
    Define \(\mathcal S^{\left(0\right)} := \{\boldsymbol 1\}\) as a vector of \(1\)'s of length \(k\) where \(k\) is the number of AR(1) regimes. For each \(t=1,2,\dots,T\), let \(\mathcal S^{\left(t\right)}\) be the set of all vectors \(\boldsymbol n_t := \left(n_{t,1},\dots,n_{t,k}\right)\) such that, for \(j,m\in \mathcal S_{AR}\), 
    \begin{enumerate}
        \item[(i)] \(n_{t,j} \in \{1,2,\dots,t+1\} \),
        \item[(ii)] there are at most \(\min\left(t,k\right)\) elements of \(\boldsymbol{n}_t\) with \(n_{t,j} \neq t+1\),
        \item[(iii)] \(n_{t,j}\neq n_{t,m}\) for all \(j\neq m\), unless \(n_{t,j}=n_{t,m}=t+1\).
    \end{enumerate}
    Given \(\{\boldsymbol{H}_t\}\) is initialised with the distribution in Equation (\ref{eqn: init distn of H_t}), it is possible for \(\boldsymbol{H}_t\) to reach states \(\left(\boldsymbol n_t,i\right)\) where \(\boldsymbol n_t \in \mathcal S^{\left(t\right)}\) and \(i \in \mathcal S\) only.
    The cardinality of \(\mathcal S^{\left(t\right)}\) is \(|\mathcal S^{\left(t\right)}| = \sum\limits_{m = 0}^{\min\left(t,k\right)}\binom{t}{m}\binom{k}{m}m!.\)
\end{lemma}
\begin{proof}
    First, we explain why \(\mathcal S^{\left(t\right)}\) contains all possible values of the counters of \(\boldsymbol{H}_t\) for \(t\geq 0\). 
    
    At time \(t=0\) the chain, \(\{\boldsymbol{H}_t\} = \{\left(\boldsymbol{N}_t,R_t\right)\}\), is initialised with the distribution in Equation (\ref{eqn: init distn of H_t}), so 
    \(\mathcal S^{\left(0\right)} := \{\boldsymbol 1\}.\)

	At time \(t>0\), either \(\{R_t\}\) has never visited state \(j\in \mathcal S_{AR}\), in which case \(n_{t,j}=t+1\), or \(\{R_t\}\) \emph{last} visited \(j\) at time \(t_j\), in which case \(n_{t,j} = t - t_j \in \{1,2,\dots,t\}\); this is part \emph{(i)}. Since the process \(\{R_t\}\) can only be in one regime at a time, it follows that \(n_{t,j}\neq n_{t,m}\) for \(j\neq m\) (unless \(n_{t,j}= n_{t,m}=t+1\)), which is part \emph{(iii)} of the definition. Also, at time \(t\), the regime chain \(\{R_t\}\) could only possibly have visited \(\min\left(t,k\right)\) possible states; this is part \emph{(ii)} of the definition. 
    
    Now, to prove the cardinality of \(\mathcal S^{\left(t\right)}\). The elements of \(\mathcal S^{\left(t\right)}\) are of the form \(\left(n_1,\dots,n_k\right)\). At time \(t\), let \(m\) be the possible number of counters that are not equal to \(t+1\), so \(m\) is an element of \(\{0,1,\dots,\min\left(t,k\right)\}\). For each\linebreak \(m \in \{0,1,\dots,\min\left(t,k\right)\}\), there are \(\binom{k}{m}\) ways of choosing which \(m\) of the \(k\) counters are not equal to \(t+1\). Next, each counter takes a distinct value in \(\{1,\dots,t\}\), so there are \(\binom{t}{m}\) ways of choosing the value of the \(m\) counters. There are \(m!\) possible permutations to allocate the chosen values to the counters. So, in total there are \(\sum\limits_{m = 0}^{\min\left(t,k\right)}\binom{t}{m}\binom{k}{m}m!\) elements in \(\mathcal S^{\left(t\right)}\).
 \end{proof}

Lemma \ref{lemma: S^t} says that if \(\boldsymbol n_t \notin \mathcal S^{\left(t\right)}\) then \(\mathbb P^{\boldsymbol{\theta}}\left(\boldsymbol{H}_t = \left(\boldsymbol{n}_t,j\right)\right) = 0\) for any \(j \in \mathcal S\). Therefore the elements of the set \(\mathcal S^{\left(t\right)}\) partition the space of all counters that the process \(\{\boldsymbol{H}_t\}\) has positive probability of reaching. Thus, for any (measurable) set \(A\) and any \(t\), the law of total probability can be applied as \linebreak
\(\mathbb P^{\boldsymbol{\theta}}\left(A\right) = \sum\limits_{\boldsymbol{n}_t\in \mathcal S^{\left(t\right)}} \sum\limits_{j \in \mathcal S} \mathbb P^{\boldsymbol{\theta}}\left(\boldsymbol{H}_t = \left(\boldsymbol{n}_t,j\right),A\right).\) We will use this fact multiple times to construct the forward algorithm. 

\subsection{Constructing the forward algorithm}
The forward algorithm is multi-purpose. It can be used to evaluate the likelihood, and also to evaluate the filtered and prediction probabilities, which in turn are inputs to the backward algorithm. 
For clarity of exposition, we first present a simple, but impractical due to underflow, algorithm (Lemma \ref{lemma: lemma 2}) to calculate the likelihood for independent regime MRS models, then address the underflow issue later with a normalised version of the algorithm (Lemma \ref{lemma: lemma 4}). Define \[\alpha_{\boldsymbol{n}_t}^{\left(t\right)}\left(j\right) := f^{\boldsymbol{\theta}}_{\boldsymbol H_t,\boldsymbol{X}_{0:t}} \left(\left(\boldsymbol n_t,j\right),\boldsymbol{x}_{0:t}\right)\] for \(t=0,1,\dots,T\), \(\boldsymbol n_t \in \mathcal S^{\left(t\right)}\), and \(j \in \mathcal S\).
\begin{lemma}[A simple forward algorithm]\label{lemma: lemma 2}
First, for \(j\in \mathcal S\) calculate  
\begin{equation}\label{eqn: hash}
\alpha^{\left(0\right)}_{\boldsymbol n_0}\left(j\right)=f^{\boldsymbol \theta}_{X_0|\boldsymbol H_0}\left(x_0|\left(\boldsymbol n_0,j\right)\right)\mathbb P^{\boldsymbol \theta}\left(\boldsymbol H_0=\left(\boldsymbol n_0,j\right)\right).
\end{equation}
Then for \(t=1,2,\dots,T\), \(\boldsymbol{n}_t\in \mathcal S^{\left(t\right)}\), \(j\in \mathcal S\), calculate
    \begin{align} \label{eqn: lemma 1}
 \alpha_{\boldsymbol{n}_t}^{\left(t\right)}\left(j\right) &=\begin{cases} 
       f^{\boldsymbol \theta}_{X_t|\boldsymbol{H}_{t},\boldsymbol{X}_{0:t-1}}\left(x_t|\left(\boldsymbol{n}_{t},j\right),\boldsymbol{x}_{0:t-1}\right) \sum\limits_{i\in \mathcal S_{AR}^c}p_{ij}\alpha_{\boldsymbol{n}_{t}-\boldsymbol 1}^{\left(t-1\right)}\left(i\right)& \\\qquad\qquad\qquad\qquad\qquad\qquad\text{if } n_{t,\ell}\neq1, \text{ for all } \ell\in \mathcal S_{AR},
        \\ f^{\boldsymbol \theta}_{X_t|\boldsymbol{H}_{t},\boldsymbol{X}_{0:t-1}}\left(x_t| \left(\boldsymbol{n}_{t},j\right),\boldsymbol{x}_{0:t-1}\right)p_{\ell j}\sum\limits_{m=1}^{t}\alpha_{\boldsymbol{n}_{t}-\boldsymbol 1 + m\boldsymbol e_\ell }^{\left(t-1\right)}\left(\ell\right) &\\\qquad\qquad\qquad\qquad\qquad\qquad \text{if }n_{t,\ell}=1, \text{ for some } \ell\in \mathcal S_{AR}. \end{cases}
    \end{align}
    Then the likelihood is given by 
    \begin{equation}\label{eqn: lemma 2}L\left(\boldsymbol{\theta}\right) = \sum\limits_{j\in \mathcal S}\sum\limits_{\boldsymbol{n}_T\in \mathcal S^{\left(T\right)}}\alpha_{\boldsymbol{n}_T}^{\left(T\right)}\left(j\right).
    \end{equation}
\end{lemma}
\begin{proof}
First, from the law of total probability, we have 
\begin{align}
	\alpha_{\boldsymbol{n}_t}^{\left(t\right)}\left(j\right) &:= f^{\boldsymbol{\theta}}_{\boldsymbol H_t,\boldsymbol{X}_{0:t}} \left(\left(\boldsymbol n_t,j\right),\boldsymbol{x}_{0:t}\right)\nonumber
	\\&= \sum_{\boldsymbol n_{t-1} \in \mathcal S^{\left(t-1\right)}}\sum_{i\in \mathcal S}f^{\boldsymbol{\theta}}_{\boldsymbol H_{t-1},\boldsymbol H_t,\boldsymbol{X}_{0:t}} \left(\left(\boldsymbol n_{t-1},i\right),\left(\boldsymbol n_t,j\right),\boldsymbol{x}_{0:t}\right).\label{eqn: proof 1}
	%
\end{align}
Now, if any \(n_{t,\ell}=1\) for \(\ell\in \{1,\dots,k\}\), then it must be that \(R_{t-1}=\ell\) and \linebreak\(\boldsymbol n_{t-1}=\boldsymbol n_{t}-\boldsymbol 1 +m\boldsymbol e_{\ell}\) for some \(m\in \{1,\dots,t\}\). Thus, in this case, the double sum in Equation (\ref{eqn: proof 1}) simplifies to 
\(\sum\limits_{m=1}^{t}f^{\boldsymbol{\theta}}_{\boldsymbol H_{t-1},\boldsymbol H_t,\boldsymbol{X}_{0:t}} \left(\left(\boldsymbol n_{t}-\boldsymbol 1 +m\boldsymbol e_{\ell},\ell\right),\left(\boldsymbol n_t,j\right),\boldsymbol{x}_{0:t}\right).\)
Otherwise, all elements of \(\boldsymbol n_{t}\) are greater than 1, in which case \(R_{t-1}\notin \mathcal S_{AR}\) and \(\boldsymbol n_{t-1}=\boldsymbol n_t-\boldsymbol 1\), so the double sum in Equation (\ref{eqn: proof 1}) simplifies to 
\[\sum\limits_{i\in \mathcal S_{AR}^c} f^{\boldsymbol{\theta}}_{\boldsymbol H_{t-1},\boldsymbol H_t,\boldsymbol{X}_{0:t}} \left(\left(\boldsymbol n_{t}-\boldsymbol 1,i\right),\left(\boldsymbol n_t,j\right),\boldsymbol{x}_{0:t}\right).\]
For both cases the following arguments are the same, so for notational convenience we will use \(\boldsymbol n\) to be either \(\boldsymbol n_{t}-\boldsymbol 1\) when \(n_{t,\ell}>1\) for all \(\ell \in \mathcal S_{AR}\), or \(\boldsymbol n_{t}-\boldsymbol 1 +m\boldsymbol e_{\ell}\) when \(n_{t,\ell}=1\) for some \(\ell \in \mathcal S_{AR}\). 

We can write the summands as 
\begin{align*}
	&f^{\boldsymbol{\theta}}_{\boldsymbol H_{t-1},\boldsymbol H_t,\boldsymbol{X}_{0:t}} \left(\left(\boldsymbol n,i\right),\left(\boldsymbol n_t,j\right),\boldsymbol{x}_{0:t}\right) 
	\\&= f^{\boldsymbol{\theta}}_{\boldsymbol H_{t-1},\boldsymbol H_t,\boldsymbol{X}_{0:t-1}} \left(\boldsymbol H_{t-1}=\left(\boldsymbol n,i\right),\boldsymbol H_{t}=\left(\boldsymbol n_t,j\right),\boldsymbol{x}_{0:t-1}\right)
	\\&\qquad{}\times f^{\boldsymbol{\theta}}_{X_t|\boldsymbol H_{t-1},\boldsymbol H_t,\boldsymbol{X}_{0:t-1}} \left(x_t|\left(\boldsymbol n,i\right),\left(\boldsymbol n_t,j\right),\boldsymbol{x}_{0:t-1}\right)
	\\&= \mathbb P^{\boldsymbol{\theta}} \left(\boldsymbol H_{t}=\left(\boldsymbol n_t,j\right)|\boldsymbol H_{t-1}=\left(\boldsymbol n,i\right),\boldsymbol{x}_{0:t-1}\right) 
	\\&\qquad{}\times f^{\boldsymbol{\theta}}_{\boldsymbol H_{t-1},\boldsymbol{X}_{0:t-1}} \left(\boldsymbol H_{t-1}=\left(\boldsymbol n,i\right),\boldsymbol{x}_{0:t-1}\right)
	\\&\qquad{}\times f^{\boldsymbol{\theta}}_{X_t|\boldsymbol H_t,\boldsymbol{X}_{0:t-1}} \left(x_t|\left(\boldsymbol n_t,j\right),\boldsymbol{x}_{0:t-1}\right),
\end{align*}
and Equation (\ref{eqn: lemma 1}) follows after noting that \[\mathbb P^{\boldsymbol{\theta}} \left(\boldsymbol H_{t}=\left(\boldsymbol n_t,j\right)|\boldsymbol H_{t-1}=\left(\boldsymbol n,i\right),\boldsymbol{x}_{0:t-1}\right) =p_{ij},\] and from the definition of \(\alpha_{\boldsymbol n}^{\left(t-1\right)}\left(i\right)\). Equation (\ref{eqn: lemma 2}) is just an application of the law of total probability. 
\end{proof}

\begin{lemma}\label{lemma: complexity for simple fwd}
	The complexity of the simple forward algorithm, as given by the total number of multiplications, is less than \( \mathcal O \left(M^2T^{k+1}k^{k}\right)\).
\end{lemma}
\begin{proof}
For \(t=0\), calculating \(\alpha^{\left(0\right)}_{\boldsymbol n_0}\left(j\right)\) by (\ref{eqn: hash}) for all \(j\in \mathcal S\) requires \(M\) multiplications in total. For each \(t\in\{1,\dots,T\}\), first consider the case where \(n_{t,\ell}\neq1\) for all \(\ell \in \mathcal S_{AR}\). Fix \(t\). Noting that \[\{\boldsymbol n_t-\boldsymbol 1: \boldsymbol n_t \in \mathcal S^{\left(t\right)}, n_{t,\ell}>1 \text{ for all } \ell \in \mathcal S\} = \mathcal S^{\left(t-1\right)},\] then there are \(|\mathcal S^{\left(t-1\right)}| \left(M-k\right) M\) multiplications to calculate all the necessary \(p_{ij}\alpha_{\boldsymbol{n}_{t}-\boldsymbol 1}^{\left(t-1\right)}\left(i\right)\) terms since \(i \in \mathcal S_{AR}^c\), \(j \in \mathcal S\) and \(\boldsymbol n_t-\boldsymbol 1\in \mathcal S^{\left(t-1\right)}\). The sum then collapses this to \(|\mathcal S^{\left(t-1\right)}| M\) terms, each of which is then multiplied by \(f^{\boldsymbol \theta}_{X_t|\boldsymbol{H}_{t},\boldsymbol{X}_{0:t-1}}\left(x_t|\left(\boldsymbol{n}_{t},j\right),\boldsymbol{x}_{0:t-1}\right)\) which requires \(|\mathcal S^{\left(t-1\right)}| M\) multiplications. 

Now consider the case where \(n_{t,\ell}=1\) for some \(\ell \in \mathcal S_{AR}\). Fix \(t\) and \(\ell\). Compute the sums \(\sum\limits_{m=1}^{t}\alpha_{\boldsymbol{n}_{t}-\boldsymbol 1+m\boldsymbol e_\ell}^{\left(t-1\right)}\left(\ell\right)\) for each \(\boldsymbol n_t\) where \(n_{t,\ell}=1\) and store them. After computing the sums, there are \({\left(|\mathcal S^{\left(t\right)}| - |\mathcal S^{\left(t-1\right)}|\right)}\) stored terms. Keep \(t\) fixed, but allow \(\ell\) to vary. Each of the stored sums is multiplied by \(p_{\ell j}\) and \(f^{\boldsymbol \theta}_{X_t|\boldsymbol{H}_{t},\boldsymbol{X}_{0:t-1}}\left(x_t| \left(\boldsymbol{n}_{t},j\right),\boldsymbol{x}_{0:t-1}\right)\) for \(\ell \in \mathcal S_{AR}\) and \(j \in \mathcal S\) which gives a total of \(2 Mk  \left(|\mathcal S^{\left(t\right)}| - |\mathcal S^{\left(t-1\right)}|\right)\) multiplications. Thus, the total number of multiplications required is 
\begin{align*}
C & = M + \sum_{t=1}^T\left[|\mathcal S^{\left(t-1\right)}| \left(M-k\right) M + |\mathcal S^{\left(t-1\right)}| M + 2 M k  \left(|\mathcal S^{\left(t\right)}| - |\mathcal S^{\left(t-1\right)}|\right)\right]
\\& \leq  M+\sum_{t=1}^T\left[\left(M^2-3Mk+M\right)|\mathcal S^{\left(t-1\right)}| + 2Mk|\mathcal S^{\left(t\right)}|\right]
\\& = M+\sum_{t=1}^T\left[\left(M^2-3Mk+M\right)\sum_{m = 0}^{\min\left(t-1,k\right)}\binom{t-1}{m}\binom{k}{m}m! \right. \\&\qquad\qquad{}\left. {}+ 2Mk\sum_{m = 0}^{\min\left(t,k\right)}\binom{t}{m}\binom{k}{m}m!\right] 
\\& \leq M+\sum_{t=1}^T\left[\left(M^2-3Mk+M\right)\sum_{m = 0}^{\min\left(t-1,k\right)}\cfrac{\left(t-1\right)^{m}}{m!}\cfrac{k^m}{m!}m! \right. \\&\qquad\qquad{}\left.{}+ 2Mk\sum_{m = 0}^{\min\left(t,k\right)}\cfrac{t^{m}}{m!}\cfrac{k^m}{m!}m!\right] 
\end{align*}
by the result \(\binom{t}{m}\leq \frac{t^m}{m!}\). This can then be bounded by 
\begin{align*}
& M+ T\left(M^2-3Mk+M\right)\sum_{m = 0}^{\min\left(T-1,k\right)}\cfrac{\left(T-1\right)^{m}}{m!}\cfrac{k^m}{m!}m! + 2TMk\sum_{m = 0}^{\min\left(T,k\right)}\cfrac{T^{m}}{m!}\cfrac{k^m}{m!}m! 
\\& \leq M+ T\left(k+1\right)\left(M^2-3Mk+M\right)\cfrac{\left(T-1\right)^{k}}{k!}\cfrac{k^k}{k!}k!  + 2TMk(k+1) \cfrac{T^{k}}{k!}\cfrac{k^k}{k!}k! 
\\& \leq  M+  \left(k+1\right)\left(M^2-3Mk+M\right)\cfrac{T\left(T-1\right)^{k}}{\left(k-1\right)!}{k^{k-1}}  + 2M\left(k+1\right)\cfrac{T^{k+1}}{\left(k-1\right)!}{k^{k}}.
\end{align*}
From which we see the complexity is bounded by \( \mathcal O \left(M^2T^{k+1}{k^k}\right)\). 
\end{proof}

To overcome possible underflow issues, we consider a normalised version of the algorithm. Define 
\[\widetilde{\alpha}_{\boldsymbol{n}_t}^{\left(t\right)}\left(j\right) := \begin{cases} 
f_{\boldsymbol{H}_0,X_0 }^{\boldsymbol\theta}\left( \left(\boldsymbol n_{0},j\right), x_0 \right) & \text{for } t=0,
\\f_{\boldsymbol{H}_t,X_t|\boldsymbol{X}_{0:t-1} }^{\boldsymbol\theta}\left( \left(\boldsymbol n_{t},j\right),x_t | \boldsymbol x_{0:t-1} \right) & \text{for }t=1,\dots T.
\end{cases}\]
\begin{lemma}[A normalised algorithm]\label{lemma: lemma 4}
	Set \(\widetilde a_{\boldsymbol n_0}^{\left(0\right)} = {\alpha}_{\boldsymbol{n}_0}^{\left(0\right)}\left(j\right)\) from the simple algorithm. Then, for \(t=1,\dots,T\) calculate 
	\begin{align*}
		\widetilde{\alpha}_{\boldsymbol{n}_t}^{\left(t\right)}\left(j\right)&:=
		\begin{cases} 
		f_{X_t|\boldsymbol{H}_t,\boldsymbol{X}_{0:t-1} }^{\boldsymbol\theta}\left( x_t | \left(\boldsymbol n_{t},j\right),\boldsymbol x_{0:t-1} \right) \sum\limits_{i\in \mathcal S_{AR}^c} p_{ij}
		\cfrac{\widetilde{\alpha}_{\boldsymbol{n}_{t}-\boldsymbol{1}}^{\left(t-1\right)}\left(i\right)}{\sum\limits_{\ell\in\mathcal S}\sum\limits_{\boldsymbol n \in \mathcal S^{\left(t-1\right)}} \widetilde{\alpha}_{\boldsymbol{n}}^{\left(t-1\right)}\left(\ell\right)} &\\\qquad\qquad\qquad\qquad\qquad\qquad\qquad\qquad \text{ if } n_{t,k}\neq1, \text{ for all } k \in \mathcal S_{AR},
		\\ f_{X_t|\boldsymbol{H}_t,\boldsymbol{X}_{0:t-1} }^{\boldsymbol\theta}\left( x_t | \left(\boldsymbol n_{t},j\right),\boldsymbol x_{0:t-1} \right) p_{ij}\sum\limits_{m=1}^t 
		\cfrac{\widetilde{\alpha}_{\boldsymbol n_{t}-\boldsymbol 1 +m\boldsymbol e_{i}}^{\left(t-1\right)}\left(i\right)}{\sum\limits_{\ell\in\mathcal S}\sum\limits_{\boldsymbol n \in \mathcal S^{\left(t-1\right)}} \widetilde{\alpha}_{\boldsymbol{n}}^{\left(t-1\right)}\left(\ell\right)} &\\\qquad\qquad\qquad\qquad\qquad\qquad\qquad\qquad \text{ otherwise}.
		\end{cases}
	\end{align*}
	Then the loglikelihood is given by 
	\begin{equation}\label{eqn: lemma 4}L(\boldsymbol \theta)=\sum_{t=0}^T\log\left(\sum_{\boldsymbol n_t\in \mathcal S^{\left(t\right)}}\sum_{i\in \mathcal S} \widetilde{\alpha}_{\boldsymbol{n}_{t}}^{\left(t\right)}\left(i\right) \right). 	\end{equation}
\end{lemma}
\begin{proof}
The definition of conditional densities gives
	\begin{align*}
		\widetilde{\alpha}_{\boldsymbol{n}_t}^{\left(t\right)}\left(j\right)&:= f_{\boldsymbol{H}_t,X_t|\boldsymbol{X}_{0:t-1} }^{\boldsymbol\theta}\left( \left(\boldsymbol n_{t},j\right),x_t | \boldsymbol x_{0:t-1} \right)
		\\&= f_{X_t|\boldsymbol{H}_t,\boldsymbol{X}_{0:t-1} }^{\boldsymbol\theta}\left( x_t | \left(\boldsymbol n_{t},j\right),\boldsymbol x_{0:t-1} \right)
		\mathbb P^{\boldsymbol\theta}\left( \boldsymbol H_t=\left(\boldsymbol n_{t},j\right)|\boldsymbol x_{0:t-1} \right),
	\end{align*}
	where 
	\begin{align}
		&\mathbb P^{\boldsymbol\theta}\left( \boldsymbol H_t=\left(\boldsymbol n_{t},j\right)|\boldsymbol x_{0:t-1} \right) \nonumber
		\\&\qquad= \sum_{\boldsymbol n_{t-1}\in \mathcal S^{\left(t-1\right)}}\sum_{i\in \mathcal S} \mathbb P^{\boldsymbol\theta}\left( \boldsymbol H_{t-1} = \left(\boldsymbol n_{t-1},i\right),\boldsymbol H_t = \left(\boldsymbol n_{t},j\right)|\boldsymbol x_{0:t-1} \right).\label{eqn: sums}
	\end{align}
	Using the same arguments as in the proof of Lemma \ref{lemma: lemma 2}, the right-hand side of (\ref{eqn: sums}) simplifies to 
	\begin{align}&\sum_{m=1}^t \mathbb P^{\boldsymbol\theta}\left( \boldsymbol H_{t-1} = \left(\boldsymbol n_{t}-\boldsymbol 1 +m\boldsymbol e_{i},i\right),\boldsymbol H_t = \left(\boldsymbol n_{t},j\right)|\boldsymbol x_{0:t-1} \right) \label{eqn: proof 3}\end{align}\(\text{when }n_{t,i}=1 \text{ for some } i \in \mathcal S_{AR}, \)
	\begin{align}&\sum_{i\in \mathcal S_{AR}^c} \mathbb P^{\boldsymbol\theta}\left( \boldsymbol H_{t-1} = \left(\boldsymbol n_{t}-\boldsymbol 1,i\right),\boldsymbol H_t = \left(\boldsymbol n_{t},j\right)|\boldsymbol x_{0:t-1} \right) \label{eqn: proof 4}\end{align}
	\( \text{when }n_{t,\ell}\neq 1 \text{ for all } \ell\in\mathcal S_{AR}.\)
	The summands in (\ref{eqn: proof 3}) and (\ref{eqn: proof 4}) can be written in the form 
	\begin{align*}
		&\mathbb P^{\boldsymbol\theta}\left(\boldsymbol H_t = \left(\boldsymbol n_{t},j\right)|\boldsymbol H_{t-1} = \left(\boldsymbol n_{t-1},i\right),\boldsymbol x_{0:t-1} \right) 
		\mathbb P^{\boldsymbol\theta}\left(\boldsymbol H_{t-1} = \left(\boldsymbol n_{t-1},i\right)|\boldsymbol x_{0:t-1} \right)
		\\& = p_{ij}				\cfrac{f_{\boldsymbol H_{t-1},X_{t-1}|\boldsymbol X_{0:t-2}}^{\boldsymbol\theta}\left( \left(\boldsymbol n_{t-1},i\right),x_{t-1}|\boldsymbol x_{0:t-2} \right)}{\sum\limits_{\ell\in\mathcal S}\sum\limits_{\boldsymbol n \in \mathcal S^{\left(t-1\right)}} f_{\boldsymbol H_{t-1},X_{t-1}|\boldsymbol X_{0:t-2}}^{\boldsymbol\theta}\left( \left(\boldsymbol n,\ell\right),x_{t-1}|\boldsymbol x_{0:t-2} \right)}
		\\& = p_{ij}
		\cfrac{\widetilde{\alpha}_{\boldsymbol{n}_{t-1}}^{\left(t-1\right)}\left(i\right)}{\sum\limits_{\ell\in\mathcal S}\sum\limits_{\boldsymbol n \in \mathcal S^{\left(t-1\right)}} \widetilde{\alpha}_{\boldsymbol{n}}^{\left(t-1\right)}\left(\ell\right)},
	\end{align*}
	which proves the result for the iterations. Equation (\ref{eqn: lemma 4}) holds from the law of total probability. 
\end{proof}
	
\begin{lemma}
	The complexity of the normalised forward algorithms, as given by the total number of multiplications, is less than \( \mathcal O \left(M^2T^{k+1}k^{k}\right)\).
\end{lemma}
\begin{proof}
The normalised algorithm is the same as the forward algorithm except that each term is divided by \(\sum\limits_{\ell\in\mathcal S}\sum\limits_{\boldsymbol n \in \mathcal S^{\left(t-1\right)}} \widetilde{\alpha}_{\boldsymbol{n}}^{\left(t-1\right)}\left(\ell\right)\). The most efficient way to do this extra step is to do the division \({p_{ij}}/{\sum\limits_{\ell\in\mathcal S}\sum\limits_{\boldsymbol n \in \mathcal S^{\left(t-1\right)}} \widetilde{\alpha}_{\boldsymbol{n}}^{\left(t-1\right)}\left(\ell\right)}\) for \(i,j \in \mathcal S\) first, which results in an additional \(TM^2\) multiplications in total. 
\end{proof}

\section{A novel backward algorithm}\label{mybkwd}
The goal of the backward algorithm is to calculate the \emph{smoothed probabilities} 
\[\gamma_{\boldsymbol{n}_t}^{\left(t\right)}\left(i\right) := \mathbb P^{\boldsymbol{\theta}}\left(\boldsymbol{H}_t = \left(\boldsymbol{n}_t,i\right)|\boldsymbol x_{0:T}\right),\]
for \(t=0,1,\dots,T\), \(\boldsymbol n_t =\left(n_{1,t},\dots,n_{k,t}\right)\in \mathcal S^{\left(t\right)}\) and \(i \in \mathcal S\). The smoothed probabilities are often of interest in their own right, but are also typically used to construct an EM algorithm. Recall that as a byproduct of the forward algorithm we obtain the filtered probabilities 
\[\widehat{\alpha}_{\boldsymbol n_t}^{\left(t\right)}\left(j\right):= \mathbb P^{\boldsymbol\theta}\left(\boldsymbol{H}_{t} = \left(\boldsymbol n_{t}, j\right) | \boldsymbol x_{0:t}\right)=\cfrac{\widetilde{\alpha}_{\boldsymbol{n}_{t}}^{\left(t\right)}\left(j\right)}{\sum\limits_{\boldsymbol n_t \in \mathcal S^{\left(t\right)}}\sum\limits_{\ell\in\mathcal S} \widetilde{\alpha}_{\boldsymbol{n}_{t}}^{\left(t\right)}\left(\ell\right)},\]
as well as the prediction probabilities 
\[\phi_{\boldsymbol n_{t}}^{\left(t\right)}\left(j\right)=\mathbb P^{\boldsymbol{\theta}}\left(\boldsymbol{H}_{t} = \left(\boldsymbol n_t,j\right) | \boldsymbol x_{0:t-1}\right),\]
for \(j\in\mathcal S\), \(\boldsymbol n_t\in \mathcal S^{\left(t\right)}\), and all \(t=0,\dots,T\). These are the inputs to the backward algorithm. 

\begin{lemma}[A backward algorithm]\label{lemma: mybkwd}
The smoothed probabilities can be calculated using the following procedure. Set
\(\gamma_{\boldsymbol{n}_T}^{\left(T\right)}\left(i\right)=\widehat{\alpha}_{\boldsymbol n_T}^{\left(T\right)}\left(i\right), \text{ for all } i \in \mathcal S,\text{ } \boldsymbol{n}_T \in \mathcal S^{\left(T\right)}.\)
Then, for \(t = T-1,T-2,\dots,0\) and for \(\boldsymbol n_t\in \mathcal S^{(t)}\) calculate
\begin{align*}
    \gamma_{\boldsymbol{n}_t}^{\left(t\right)}\left(i\right) =& \begin{cases} 
    \widehat{\alpha}_{\boldsymbol{n}_t}^{\left(t\right)}\left(i\right) \sum\limits_{j \in \mathcal S} p_{ij}\cfrac{\gamma_{\boldsymbol{n}_t+\boldsymbol{1}}^{\left(t+1\right)}\left(j\right)}{\phi_{\boldsymbol{n}_t+\boldsymbol 1}^{\left(t+1\right)}\left(j\right)} & \text{for } i \in  \mathcal S_{AR}^c,
    \\ \widehat{\alpha}_{\boldsymbol{n}_t}^{\left(t\right)}\left(i\right) \sum\limits_{j \in\mathcal S} p_{ij}\cfrac{\gamma_{\boldsymbol{n}_t^{\left(-i\right)}+\boldsymbol{1}}^{\left(t+1\right)}\left(j\right)}{\phi_{\boldsymbol{n}_t^{\left(-i\right)}+\boldsymbol{1}}^{\left(t+1\right)}\left(j\right)}  & \text{for } i \in \mathcal S_{AR}.
    \end{cases}
\end{align*}
\end{lemma}
\begin{proof}
Consider the event \(\boldsymbol{H}_t = \left(\boldsymbol n_t,i\right)\). By the definition of \(\boldsymbol{n}_t\), when \(i \in \mathcal S_{AR}\), then \(\boldsymbol{n}_{t+1}=\boldsymbol{n}_{t}^{\left(-i\right)} + \boldsymbol 1\), and when \(i \in \mathcal S_{AR}^c\), then \(\boldsymbol{n}_{t+1}=\boldsymbol{n}_{t} + \boldsymbol 1\). Thus, when \(\boldsymbol{H}_t=\left(\boldsymbol{n}_t,i\right)\) is known, then \(\boldsymbol{n}_{t+1}\) is also known. As a result, 
\begin{align}
\gamma_{\boldsymbol{n}_t}^{\left(t\right)}\left(i\right)\nonumber&=\mathbb P^{\boldsymbol{\theta}}\left(\boldsymbol{H}_t = \left(\boldsymbol n_t,i\right) | \boldsymbol x_{0:T}\right) \nonumber
\\&= \sum_{j \in \mathcal S } \sum_{\boldsymbol n_{t+1} \in \mathcal S^{\left(t+1\right)}} \mathbb P^{\boldsymbol{\theta}}\left(\boldsymbol{H}_t = \left(\boldsymbol n_t,i\right), \boldsymbol{H}_{t+1} = \left(\boldsymbol n_{t+1},j\right) | \boldsymbol x_{0:T}\right) \nonumber
\\
&= \begin{cases}
\sum\limits_{j \in \mathcal S} \mathbb P^{\boldsymbol{\theta}}\left(\boldsymbol{H}_t = \left(\boldsymbol n_t,i\right), \boldsymbol{H}_{t+1} = \left(\boldsymbol n_{t}^{\left(-i\right)} +\boldsymbol{1},j\right) | \boldsymbol x_{0:T}\right) & \text{for }i \in \mathcal S_{AR},
\\ \sum\limits_{j \in \mathcal S} \mathbb P^{\boldsymbol{\theta}}\left(\boldsymbol{H}_t = \left(\boldsymbol n_t,i\right), \boldsymbol{H}_{t+1} = \left(\boldsymbol n_{t} +\boldsymbol{1},j\right) | \boldsymbol x_{0:T}\right) & \text{for }i \in \mathcal S_{AR}^c,
\end{cases} \label{eqn: piecewise Gamma 1}
\end{align}
for \(t=0,1,\dots,T-1\), \(\boldsymbol n_t \in \mathcal S^{\left(t\right)}\) and \(i \in \mathcal S\). Since the following arguments are the same for both cases, \(i \in \mathcal S_{AR}\) and \(i \in \mathcal S_{AR}^c\), for notational convenience, let \(\boldsymbol{n}\) take the value \(\boldsymbol n_{t}^{\left(-i\right)} +\boldsymbol{1}\) when \(i \in \mathcal S_{AR}\) and the value \(\boldsymbol n_{t} +\boldsymbol{1}\) when \(i \in \mathcal S_{AR}^c\). The summands on the right hand side of (\ref{eqn: piecewise Gamma 1}) can be written as
\begin{align}
%
& \cfrac{f^{\boldsymbol{\theta}}_{\boldsymbol{H}_t , \boldsymbol{H}_{t+1}, \boldsymbol X_{t+1:T} | \boldsymbol X_{0:t}}\left(\left(\boldsymbol n_t,i\right), \left(\boldsymbol n,j\right), \boldsymbol x_{t+1:T} | \boldsymbol x_{0:t}\right)}{f_{\boldsymbol X_{t+1:T} | \boldsymbol X_{0:t}}^{\boldsymbol{\theta}}\left(\boldsymbol x_{t+1:T} | \boldsymbol x_{0:t}\right)} \nonumber
\\&=  \mathbb P^{\boldsymbol{\theta}}\left(\boldsymbol{H}_t = \left(\boldsymbol n_t,i\right)|\boldsymbol x_{0:t}\right) \mathbb P^{\boldsymbol{\theta}}\left(\boldsymbol{H}_{t+1} = \left(\boldsymbol n ,j\right)|\boldsymbol{H}_t = \left(\boldsymbol n_t,i\right), \boldsymbol x_{0:t}\right) \nonumber
\\&\qquad{}\times\cfrac{f^{\boldsymbol{\theta}}_{\boldsymbol X_{t+1:T}|\boldsymbol{H}_t, \boldsymbol{H}_{t+1} , \boldsymbol X_{0:t}}\left(\boldsymbol x_{t+1:T}|\left(\boldsymbol n_t,i\right), \left(\boldsymbol n,j\right), \boldsymbol x_{0:t}\right)}{f_{\boldsymbol X_{t+1:T} | \boldsymbol X_{0:t}}^{\boldsymbol{\theta}}\left(\boldsymbol x_{t+1:T} | \boldsymbol x_{0:t}\right)} \nonumber
\\&= \mathbb P^{\boldsymbol{\theta}}\left(\boldsymbol{H}_t = \left(\boldsymbol n_t,i\right)|\boldsymbol x_{0:t}\right) p_{ij} \nonumber
\\&\qquad{}\times\cfrac{f^{\boldsymbol{\theta}}_{\boldsymbol X_{t+1:T}|\boldsymbol{H}_t , \boldsymbol{H}_{t+1} , \boldsymbol X_{0:t}}\left(\boldsymbol x_{t+1:T}|\left(\boldsymbol n_t,i\right), \left(\boldsymbol n ,j\right), \boldsymbol x_{0:t}\right)}{f_{\boldsymbol X_{t+1:T} | \boldsymbol X_{0:t}}^{\boldsymbol{\theta}}\left(\boldsymbol x_{t+1:T} | \boldsymbol x_{0:t}\right)}, \label{eqn: a statement to continue 1}
\end{align}
where the last equality holds since \(\boldsymbol{H}_{t+1}\) is independent of \(\boldsymbol{x}_{0:t}\) given \(\boldsymbol{H}_t\).
Now, noting that \(\boldsymbol x_{t+1:T}\) is independent of \(\boldsymbol{H}_t\) given \(\boldsymbol{H}_{t+1}\) and \(\boldsymbol x_{0:t}\), then the right-hand side of (\ref{eqn: a statement to continue 1}) equals 
\begin{align*}
& \cfrac{\mathbb P^{\boldsymbol{\theta}}\left(\boldsymbol{H}_t = \left(\boldsymbol n_t,i\right)|\boldsymbol x_{0:t}\right) p_{ij}}{f^{\boldsymbol{\theta}}_{\boldsymbol X_{t+1:T} | \boldsymbol X_{0:t}}\left(\boldsymbol x_{t+1:T} | \boldsymbol x_{0:t}\right)} 
f^{\boldsymbol{\theta}}_{\boldsymbol X_{t+1:T}| \boldsymbol{H}_{t+1}, \boldsymbol X_{0:t}}\left(\boldsymbol x_{t+1:T}|\left(\boldsymbol n ,j\right), \boldsymbol x_{0:t}\right),
\\&= \cfrac{\mathbb P^{\boldsymbol{\theta}}\left(\boldsymbol{H}_t = \left(\boldsymbol n_t,i\right)|\boldsymbol x_{0:t}\right) p_{ij}}{f^{\boldsymbol{\theta}}_{\boldsymbol X_{t+1:T} | \boldsymbol X_{0:t}}\left(\boldsymbol x_{t+1:T} | \boldsymbol x_{0:t}\right)} 
\cfrac{f^{\boldsymbol{\theta}}_{\boldsymbol X_{t+1:T}, \boldsymbol{H}_{t+1}| \boldsymbol X_{0:t}}\left(\boldsymbol x_{t+1:T}, \boldsymbol{H}_{t+1} = \left(\boldsymbol n ,j\right)| \boldsymbol x_{0:t}\right)}{\mathbb P^{\boldsymbol{\theta}}\left(\boldsymbol{H}_{t+1} = \left(\boldsymbol n ,j\right) | \boldsymbol x_{0:t}\right)}  \nonumber
\\&= \cfrac{\mathbb P^{\boldsymbol{\theta}}\left(\boldsymbol{H}_t = \left(\boldsymbol n_t,i\right)|\boldsymbol x_{0:t}\right) p_{ij}}{f^{\boldsymbol{\theta}}_{\boldsymbol X_{t+1:T} | \boldsymbol X_{0:t}}\left(\boldsymbol x_{t+1:T} | \boldsymbol x_{0:t}\right)}\nonumber
\cfrac{f^{\boldsymbol{\theta}}_{\boldsymbol X_{t+1:T} | \boldsymbol X_{0:t}}\left(\boldsymbol x_{t+1:T} | \boldsymbol x_{0:t}\right)\mathbb P^{\boldsymbol{\theta}}\left(\boldsymbol{H}_{t+1} = \left(\boldsymbol n ,j\right)| \boldsymbol x_{0:T}\right)}{\mathbb P^{\boldsymbol{\theta}}\left(\boldsymbol{H}_{t+1} = \left(\boldsymbol n,j\right) | \boldsymbol x_{0:t}\right)}  \nonumber
\\&= \mathbb P^{\boldsymbol{\theta}}\left(\boldsymbol{H}_t = \left(\boldsymbol n_t,i\right)|\boldsymbol x_{0:t}\right) p_{ij} 
\cfrac{\mathbb P^{\boldsymbol{\theta}}\left(\boldsymbol{H}_{t+1} = \left(\boldsymbol n ,j\right)| \boldsymbol x_{0:T}\right)}{\mathbb P^{\boldsymbol{\theta}}\left(\boldsymbol{H}_{t+1} = \left(\boldsymbol n,j\right) | \boldsymbol x_{0:t}\right)}  \nonumber 
= \widehat{\alpha}_{\boldsymbol{n}_t}^{\left(t\right)}\left(i\right)  p_{ij} \cfrac{\gamma_{\boldsymbol{n}}^{\left(t\right)}\left(j\right)}{\phi_{\boldsymbol{n} }^{\left(t\right)}\left(j\right)}.
\end{align*}

Writing out \(\boldsymbol{n}\) explicitly for the two cases completes the proof.
\end{proof}

\begin{lemma}
The total complexity of the backward algorithm in Lemma \ref{lemma: mybkwd}, as measured by the total number of multiplications, is less than \(\mathcal O\left(M^2 T^{k+1}k^k\right).\)
\end{lemma}
\begin{proof}
First, for each \(t\in \{T-1,...,0\}\) we need to calculate the ratio \[\cfrac{\gamma_{\boldsymbol{n}_{t+1}}^{\left(t+1\right)}\left(j\right)}{\phi_{\boldsymbol{n}_{t+1}}^{\left(t+1\right)}\left(j\right)}\] for every corresponding \(\boldsymbol n_t\in \mathcal S^{\left(t\right)}\) and \(j \in \mathcal S\). This costs \(M|\mathcal S^{\left(t\right)}|\) multiplications. This quantity is independent of \(i\), thus only needs to be done once for a given \(t\) if we save the results. 

Now consider \(t\), \(i\) and \(\boldsymbol n_t\) fixed. The multiplication of \(p_{ij}\) and \[\cfrac{\gamma_{\boldsymbol{n}_{t+1}}^{\left(t+1\right)}\left(j\right)}{\phi_{\boldsymbol{n}_{t+1}}^{\left(t+1\right)}\left(j\right)}\] is done for every \(j\in \mathcal S\) which costs \(M\) multiplications. The sum over \(j\in \mathcal S\) results in a single term, which is then multiplied by the corresponding \( \widehat{\alpha}_{\boldsymbol{n}_t}^{\left(t\right)}\left(i\right) \), and this costs an additional 1 multiplication. We do this for all \(i\in \mathcal S\) and \(\boldsymbol n_t\in \mathcal S^{\left(t\right)}\), which costs \(\left(M+1\right)M|\mathcal S^{\left(t\right)}|\) multiplications. 

So, for a given \(t\) we execute \(M|\mathcal S^{\left(t\right)}|+\left(M+1\right)M|\mathcal S^{\left(t\right)}|\) multiplications. This is done for every \(t=0,\dots,T-1\), so the total number of multiplications is
\begin{align*}
\sum_{t=0}^{T-1}\left(M|\mathcal S^{\left(t\right)}|+\left(M+1\right)M|\mathcal S^{\left(t\right)}| \right)&= \left(M^2+2M\right)\sum_{t=0}^{T-1}|\mathcal S^{\left(t\right)}|
\\&\leq\left(M^2+2M\right)\cfrac{T\left(T-1\right)^{k}k^{k-1}\left(k+1\right)}{\left(k-1\right)!}
\\&=\mathcal O\left(M^2T^{k+1}k^k\right),
\end{align*}
where we have used similar arguments to Lemma \ref{lemma: complexity for simple fwd} to bound the complexity. \end{proof}


Of importance to the next section, note that we can obtain from the backward algorithm, the smoothed probabilities 
\[\mathbb P^{\boldsymbol{\theta}}\left(R_t = i, N_{t,i} = \ell |\boldsymbol{x}_{0:T}\right) = \sum\limits_{\substack{\boldsymbol n_t \in \mathcal S^{\left(t\right)}: \\ 
    n_{t,i} = \ell }} \mathbb P^{\boldsymbol{\theta}}\left(\boldsymbol{H}_t = \left(\boldsymbol n_t,i\right)|\boldsymbol{x}_{0:T}\right).\]

\section{A novel EM algorithm}\label{myEM}
Here we show how the output from the backward algorithm can be used to implement an exact, computationally feasible EM algorithm for MRS models with independent regimes. 

\subsection{The E-step}\label{myEstep}
Recall that the EM algorithm \citep{dempster1977} is an iterative procedure, alternating between an expectation step and a maximisation step. In the expectation step the function \(Q\left(\boldsymbol \theta, \boldsymbol \theta_n\right)\) is constructed as 
\begin{align} 
Q\left(\boldsymbol \theta, \boldsymbol \theta_n\right) &= \mathbb E [\log f_{\boldsymbol X_{0:T}, \boldsymbol R}^{\boldsymbol\theta} \left(\boldsymbol x_{0:T}, \boldsymbol R\right)|\boldsymbol{x}_{0:T};\boldsymbol{\theta}_n]\nonumber\\&= \mathbb E [\log f_{\boldsymbol X_{0:T}, \boldsymbol H_0,\dots,\boldsymbol{H}_T}^{\boldsymbol\theta} \left(\boldsymbol x_{0:T}, \boldsymbol H_0,\dots,\boldsymbol{H}_T\right)|\boldsymbol{x}_{0:T}; \boldsymbol{\theta}_n],\label{fullEMeqn0}
\end{align}
where \(\boldsymbol R = \left(R_0,\dots,R_T\right)\) is a sequence of the hidden Markov chain \(\{R_t\}\), and \(\left(\boldsymbol H_0,\dots,\boldsymbol{H}_T\right)\) is a sequence of the corresponding augmented hidden process \(\{\boldsymbol{H}_t\}\). The information contained in the sequences \(\boldsymbol R\) and \(\left(\boldsymbol H_0,\dots,\boldsymbol{H}_T\right)\) is entirely equivalent, but we opt for the latter representation to remain consistent with, and emphasise the place of, the work in the previous sections. In the M-step of the algorithm, the maximisers \(\argmax\limits_{\boldsymbol{\theta}\in \Theta} Q\left(\boldsymbol{\theta}, \boldsymbol{\theta}_n\right)\) are found. 

For MRS models \(Q\left(\boldsymbol{\theta},\boldsymbol{\theta}_n\right)\) can be written as 
\begin{align} 
Q\left(\boldsymbol \theta, \boldsymbol \theta_n\right) 
&= \mathbb E\left[ \log f_{\boldsymbol X_{0:T}| \boldsymbol H_0,\dots,\boldsymbol{H}_T}^{\boldsymbol{\theta}}\left(\boldsymbol{x}_{0:T}|\boldsymbol H_0,\dots,\boldsymbol{H}_T\right)|\boldsymbol{x}_{0:T};\boldsymbol{\theta}_n\right]\nonumber \\&\qquad{}+ \mathbb E\left[ \log \mathbb{P}^{\boldsymbol{\theta}}\left(\boldsymbol H_0,\dots,\boldsymbol{H}_T\right)|\boldsymbol{x}_{0:T}; \boldsymbol{\theta}_n\right]. \label{fullEMeqn1}
\end{align}
Using the augmented hidden Markov chain, \(\{\boldsymbol{H}_t\}_{t\in \mathbb N}\), (\ref{fullEMeqn1}) can be written in such a way that the function \(Q\) is computationally feasible. First note that, given \(\boldsymbol{H}_t\) and \(\boldsymbol X_{0:t-1}\), \(X_t\) is independent of \(\boldsymbol{H}_s\) for \(s \neq t\) which allows the function \(\log f^{\boldsymbol\theta}_{\boldsymbol X_{0:T}| \boldsymbol H_0,\dots,\boldsymbol H_T} \left(\boldsymbol x_{0:T}| \boldsymbol H_0,\dots,\boldsymbol H_T\right)\) to be written as
\begin{align}
&\log f^{\boldsymbol\theta}_{\boldsymbol X_{0:T}| \boldsymbol H_0,\dots,\boldsymbol{H}_T} \left(\boldsymbol x_{0:T}| \boldsymbol H_0,\dots,\boldsymbol{H}_T\right)  \nonumber
%
%
\\&= \log f^{\boldsymbol \theta}_{X_0|\boldsymbol{H}_0 }\left(x_0|\boldsymbol{H}_0 \right) + \sum_{t = 1}^T\log f^{\boldsymbol\theta}_{X_t|\boldsymbol{H}_t,\boldsymbol{X}_{0:t-1}} \left( x_t| \boldsymbol{H}_t,\boldsymbol{x}_{0:t-1} \right) \nonumber
\\&=\log \left\{\prod_{j \in \mathcal S}\prod_{\boldsymbol n_0 \in \mathcal S^{\left(0\right)}} f^{\boldsymbol \theta}_{X_0|\boldsymbol{H}_0}\left( x_0|\left(\boldsymbol{n}_0,j\right) \right)^{\mathbb I \left( \boldsymbol{H}_0 = \left(\boldsymbol n_{0}, j\right)\right)}\right\}  
\\&\qquad{}+  \sum_{t = 1}^T\log\left\{ \prod_{j \in \mathcal S}\prod_{\boldsymbol n_t \in \mathcal S^{\left(t\right)} } f_{X_t|\boldsymbol{H}_t, \boldsymbol{X}_{0:t-1}}^{\boldsymbol\theta} \left( x_t| \left(\boldsymbol{n}_t,j\right), \boldsymbol{x}_{0:t-1}\right)^{\mathbb I \left(\boldsymbol{H}_t = \left(\boldsymbol{n}_t,j\right)\right)} \right\}  \nonumber
%
%
%
\\ &=  \sum_{j \in \mathcal S}\sum_{\boldsymbol n_0 \in \mathcal S^{\left(0\right)}}{\mathbb I \left( \boldsymbol{H}_0 = \left(\boldsymbol{n}_0,j\right)\right)} \log f^{\boldsymbol \theta}_{X_0|\boldsymbol{H}_0}\left(x_0| \left(\boldsymbol{n}_0,j\right)\right)  \nonumber
\\&\qquad{}+ \sum_{t = 1}^T  \sum_{j \in \mathcal S}\sum_{\boldsymbol n_t \in \mathcal S^{\left(t\right)} } {\mathbb I \left(\boldsymbol{H}_t = \left(\boldsymbol{n}_t,j\right)\right)} \log  f^{\boldsymbol\theta}_{X_t|\boldsymbol{H}_t , \boldsymbol{X}_{0:t-1} }\left( x_t| \left(\boldsymbol{n}_t,j\right), \boldsymbol{x}_{0:t-1} \right) . \label{eqn: this simplifies to}
\end{align}
Since \(f^{\boldsymbol\theta}_{X_t|\boldsymbol{H}_t , \boldsymbol{X}_{0:t-1}}\left( x_t|  \left(\boldsymbol{n}_t,j\right), \boldsymbol{x}_{0:t-1} \right)  = f^{\boldsymbol\theta}_{X_t|N_{t,j},R_t,\boldsymbol X_{0:t-1}}\left( x_t | n_{t,j}, j, \boldsymbol{x}_{0:t-1} \right) \), and similarly for \(f^{\boldsymbol{\theta}}_{X_0|\boldsymbol{H_0}}\left(x_0|\left(\boldsymbol{n}_{t,j},j\right)\right)\), the expression (\ref{eqn: this simplifies to}) simplifies to
\begin{align}
&\sum_{j \in \mathcal S_{AR}}{\mathbb I \left( N_{0,j} = 1, R_0 = j \right)} \log  f^{\boldsymbol \theta}_{X_0| N_{0,j} ,R_0  }\left(x_0| 1, j \right) 
\\&{}+ \sum_{t = 1}^T  \sum_{j \in \mathcal S_{AR}}\sum_{ m =1 }^{t} {\mathbb I \left( N_{t,j} = m,R_t = j\right)} \log f^{\boldsymbol\theta}_{ X_t | N_{t,j} , R_t,\boldsymbol{X}_{0:t-1}}\left( x_t | m , j,\boldsymbol{x}_{0:t-1} \right) \nonumber
\\& + \sum_{j \in \mathcal S_{AR}^c}{\mathbb I \left( R_0 = j \right)} \log  f^{\boldsymbol \theta}_{X_0| R_0  }\left(x_0| j \right) \nonumber
\\&+ \sum_{t = 1}^T  \sum_{j \in \mathcal S_{AR}^c} {\mathbb I \left( R_t = j\right)} \log f^{\boldsymbol\theta}_{ X_t | R_t}\left( x_t | j\right). \label{takeExpofthis}
\end{align}
Taking the expectation of (\ref{takeExpofthis}) with respect to the distribution \(f^{\boldsymbol{\theta}_n}_{\boldsymbol H_0,\dots,\boldsymbol{H}_T|\boldsymbol{X}_{0:T}}\) (equivalently the distribution \(f^{\boldsymbol{\theta}_n}_{\boldsymbol{R}|\boldsymbol{X}_{0:T}}\)) gives 
\begin{align*}
    & \mathbb E \left[\log f^{\boldsymbol\theta}_{\boldsymbol X_{0:T}, \boldsymbol H_0,\dots,\boldsymbol{H}_T|\boldsymbol{X}_{0:T}} \left(\boldsymbol x_{0:T}, \boldsymbol H_0,\dots,\boldsymbol{H}_T\right)|\boldsymbol{x}_{0:T}; \boldsymbol{\theta}_n\right] \nonumber
    \\&=  \sum_{j \in \mathcal S_{AR}} \mathbb P^{\boldsymbol \theta_n} \left(  N_{0,j} = 1 , R_0 = j |\boldsymbol{x}_{0:T}\right) \log f^{\boldsymbol \theta}_{X_0| N_{0,j},  R_0 }\left(x_0| 1 , j \right)
    \\&{}  + \sum_{t = 1}^T  \sum_{j \in \mathcal S_{AR}}\sum_{m=1}^t {\mathbb P^{\boldsymbol \theta_n} \left( N_{t,j} = m, R_t = j |\boldsymbol{x}_{0:T}\right) } \log f^{\boldsymbol\theta}_{X_t|N_{t,j},R_t,\boldsymbol X_{0:t-1}} \left( x_t| m, j, \boldsymbol{x}_{0:t-1} \right) 
    \\&{}  +\sum_{j \in \mathcal S_{AR}^c} {\mathbb P^{\boldsymbol \theta_n} \left(  R_0 = j  |\boldsymbol{x}_{0:T}\right)} \log  f^{\boldsymbol \theta}_{X_0|R_0}\left(x_0| j \right)  
    \\&{} + \sum_{t = 1}^T  \sum_{j \in \mathcal S_{AR}^c}  {\mathbb P^{\boldsymbol \theta_n} \left(   R_t = j |\boldsymbol{x}_{0:T}\right) } \log  f^{\boldsymbol\theta}_{X_t|R_t} \left( x_t| j  \right).
\end{align*}

Using similar arguments, \(\mathbb E\left[\log\mathbb P^{\boldsymbol \theta} \left(\boldsymbol H_0,\dots,\boldsymbol{H}_T\right)|\boldsymbol x_{0:T}; \boldsymbol{\theta}_n\right]\) is found to be
\begin{align*}
    &\mathbb E\left[\log\mathbb P^{\boldsymbol \theta} \left(\boldsymbol H_0,\dots,\boldsymbol{H}_T\right)  \bigg|\boldsymbol x_{0:T};\boldsymbol{\theta}_n\right] 
   \\ &= \mathbb E\left[\log \left\{\prod_{i\in \mathcal S}\pi_i^{\mathbb I \left(R_0=i\right)}\prod_{i,j \in \mathcal S}p_{ij}^{\eta_{ij}}\right\}\bigg|\boldsymbol x_{0:T};\boldsymbol{\theta}_n\right] 
    %
    %
    \\&= \sum_{i \in \mathcal S} \mathbb P^{\boldsymbol \theta_n} \left(R_0 = i | \boldsymbol x_{0:T}\right)\log \pi_i + \sum_{i,j \in \mathcal S} \mathbb E\left[\eta_{ij}|\boldsymbol x_{0:T};\boldsymbol{\theta}_n\right] \log p_{ij},
\end{align*}
where \(\eta_{ij}\) is the random variable counting the number of transitions from state \(R_{t-1} = i\) to state \(R_t = j\) in the sequence \(\boldsymbol{R} \). The expectation \(\mathbb E\left[\eta_{ij}|\boldsymbol x_{0:T};\boldsymbol{\theta}_n\right]\) can be calculated as 
\begin{align*}
    \mathbb E\left[\eta_{ij}|\boldsymbol x_{0:T};\boldsymbol{\theta}_n\right] & = \mathbb E\left[\sum_{t=1}^T\mathbb I\left(R_{t-1} = i, R_t = j\right)\bigg|\boldsymbol x_{0:T};\boldsymbol{\theta}_n\right]
    %
    %
   \\& = \sum_{t=1}^T\mathbb P^{\boldsymbol{\theta}_n}\left(R_{t-1} = i, R_t = j|\boldsymbol x_{0:T}\right).
\end{align*}
So, the function \(Q\) is
\begin{align}
    &Q\left(\boldsymbol \theta, \boldsymbol \theta_n\right) \nonumber
    %
    %
    \\&=\sum_{j \in \mathcal S_{AR}}{\mathbb P^{\boldsymbol \theta_n} \left( N_{0,j} = 1, R_0 = j |\boldsymbol{x}_{0:T}\right)} \log f^{\boldsymbol \theta}_{X_0|N_{0,j},R_0}\left(x_0| 1 , j \right) \nonumber
    \\&{}+\sum_{j \in \mathcal S_{AR}^c} {\mathbb P^{\boldsymbol \theta_n} \left(  R_0 = j  |\boldsymbol{x}_{0:T}\right)} \log  f^{\boldsymbol \theta}_{X_0|R_0}\left(x_0| j \right) \nonumber
    \\& {} + \sum_{t = 1}^T  \sum_{j \in \mathcal S_{AR}}\sum_{m=1 }^t {\mathbb P^{\boldsymbol \theta_n} \left( N_{t,j} = m ,  R_t = j |\boldsymbol{x}_{0:T}\right) } \log f^{\boldsymbol\theta}_{X_t|N_{t,j},R_t,\boldsymbol X_{0:t-1}} \left( x_t| m, j, \boldsymbol{x}_{0:t-1} \right) \nonumber
    \\& {} + \sum_{t = 1}^T  \sum_{j \in \mathcal S_{AR}^c}  {\mathbb P^{\boldsymbol \theta_n} \left(   R_t = j |\boldsymbol{x}_{0:T}\right) } \log  f^{\boldsymbol\theta}_{X_t|R_t} \left( x_t| j \right) \nonumber
    \\& {} + \sum_{i \in \mathcal S} \mathbb P^{\boldsymbol \theta_n} \left(R_0 = i | \boldsymbol x_{0:T}\right)\log \pi_i + \sum_{i,j \in \mathcal S}\sum_{t=1}^T\mathbb P^{\boldsymbol{\theta}_n}\left(R_{t-1} = i, R_t = j|\boldsymbol x_{0:T}\right)\log p_{ij}.\label{eqn: Q function simple MRS}
    %
    %
    %
\end{align}

\begin{lemma}\label{lemma: joint probs}
    The joint probabilities are given by 
\begin{align}
    &\mathbb P^{\boldsymbol{\theta}_n}\left(R_{t-1} = i, R_t = j|\boldsymbol x_{0:T}\right)\nonumber
    \\&= \begin{cases}
         \mathbb P^{\boldsymbol{\theta}_n}\left(N_{t,i} = 1, R_t = j|\boldsymbol x_{0:T}\right) 
	\qquad\qquad\qquad\qquad\qquad\qquad\quad   \text{ when } i \in \mathcal S_{AR}, 
    \\  \sum\limits_{\boldsymbol n_t-\boldsymbol{1} \in \mathcal S^{\left(t-1\right)}} \mathbb P^{\boldsymbol{\theta}_n}\left(R_t = j, \boldsymbol{N}_{t} = \boldsymbol{n}_{t}|\boldsymbol x_{0:T}\right) \\{}\qquad\times\cfrac{ p_{ij}^{\left(n\right)}\mathbb P^{\boldsymbol{\theta}_n}\left(R_{t-1} = i,\boldsymbol{N}_{t-1} =\boldsymbol{n}_{t}-\boldsymbol{1}| \boldsymbol x_{0:t-1}\right)}{\sum\limits_{k \in \mathcal S_{AR}^c}p_{k j}^{\left(n\right)}\mathbb P^{\boldsymbol{\theta}_n}\left(R_{t-1}=k,\boldsymbol{N}_{t-1} =\boldsymbol{n}_{t}-\boldsymbol{1}|\boldsymbol{x}_{0:t-1}\right)}, \label{eqn: take care here}
\text{ }\text{ when } i \in \mathcal S_{AR}^c.
	\end{cases}
\end{align}
\end{lemma}
This proof follows similar arguments to those in \cite{kim1994}, which develops algorithms for MRS models with dependent regimes. 
\begin{proof}
For the case \(i \in \mathcal S_{AR}\), note that \(N_{t,i}=1\) if and only if \(R_{t-1}=i\) and we are done.

When \(i\in \mathcal S_{AR}^c\) all counters in \(\boldsymbol{n}_t\) are different from 1, so \(\boldsymbol{n}_t-\boldsymbol{1}\in \mathcal S^{\left(t-1\right)}\). Thus
\begin{align}
    &\mathbb P^{\boldsymbol{\theta}_n}\left(R_{t-1} = i, R_t = j|\boldsymbol x_{0:T}\right) \nonumber
\\    &=\sum\limits_{\boldsymbol n_{t}-\boldsymbol{1} \in \mathcal S^{\left(t-1\right)}} \mathbb P^{\boldsymbol{\theta}_n}\left( \boldsymbol{N}_{t} = \boldsymbol{n}_{t}, R_{t-1} = i, R_t = j|\boldsymbol x_{0:T}\right)\nonumber
    \\&= \sum\limits_{\boldsymbol n_t-\boldsymbol{1} \in \mathcal S^{\left(t-1\right)}} \mathbb P^{\boldsymbol{\theta}_n}\left(\boldsymbol{N}_{t} = \boldsymbol{n}_{t}, R_t = j |\boldsymbol x_{0:T}\right)\mathbb P^{\boldsymbol{\theta}_n}\left(R_{t-1} = i| \boldsymbol{N}_{t} = \boldsymbol{n}_{t}, R_t = j, \boldsymbol x_{0:T}\right)\nonumber
    \\&= \sum\limits_{\boldsymbol n_t -\boldsymbol{1}\in \mathcal S^{\left(t-1\right)}} \mathbb P^{\boldsymbol{\theta}_n}\left( \boldsymbol{N}_{t} = \boldsymbol{n}_{t} , R_t = j |\boldsymbol x_{0:T}\right)\nonumber\\&\qquad{}\times\mathbb P^{\boldsymbol{\theta}_n}\left(R_{t-1} = i| \boldsymbol{N}_{t} = \boldsymbol{n}_{t}, R_t = j,  \boldsymbol x_{0:t-1}\right).\label{eqn: continue from here 1}
\end{align}
The last equality 
holds since, 
    %
    %
given \(R_t\) and \(\boldsymbol{N}_t\), then \(\boldsymbol{x}_{t:T}\) is independent of \(R_{t-1}\). Focusing on the right-most term in Equation (\ref{eqn: continue from here 1}),
\begin{align*}
    &\mathbb P^{\boldsymbol{\theta}_n}\left(R_{t-1} = i| \boldsymbol{N}_{t} = \boldsymbol{n}_{t}, R_t = j, \boldsymbol x_{0:t-1}\right)\nonumber
    \\&= \cfrac{\mathbb P^{\boldsymbol{\theta}_n}\left(R_{t} = j| \boldsymbol{N}_{t} =\boldsymbol{n}_{t}, R_{t-1}=i, \boldsymbol x_{0:t-1}\right)\mathbb P^{\boldsymbol{\theta}_n}\left(R_{t-1} = i|\boldsymbol{N}_{t} =\boldsymbol{n}_{t}, \boldsymbol x_{0:t-1}\right)}{\mathbb P^{\boldsymbol{\theta}_n}\left(R_t=j|\boldsymbol{N}_{t} =\boldsymbol{n}_{t},\boldsymbol{x}_{0:t-1}\right)}\nonumber
    \\&= \cfrac{ p_{ij}^{\left(n\right)}\mathbb P^{\boldsymbol{\theta}_n}\left(\boldsymbol{N}_{t} =\boldsymbol{n}_{t}, R_{t-1} = i | \boldsymbol x_{0:t-1}\right)}{\mathbb P^{\boldsymbol{\theta}_n}\left(\boldsymbol{N}_{t} =\boldsymbol{n}_{t}, R_t=j |\boldsymbol{x}_{0:t-1}\right)},
\end{align*}
where \(p_{ij}^{\left(n\right)}\) is the parameter \(p_{ij}\) in \(\boldsymbol{\theta}_n\); the second equality holds since, given \(R_{t-1}\), then \(R_t\) is independent of \(\boldsymbol{N}_t\) and \(\boldsymbol{X}_{0:t-1}\). 

Now, notice that 
\[\mathbb P^{\boldsymbol{\theta}_n}\left( \boldsymbol{N}_{t} =\boldsymbol{n}_{t} , R_{t-1} = i | \boldsymbol x_{0:t-1}\right) = \mathbb P^{\boldsymbol{\theta}_n}\left( \boldsymbol{N}_{t-1} =\boldsymbol{n}_{t}-\boldsymbol{1} , R_{t-1} = i | \boldsymbol x_{0:t-1}\right),\]
since \(i \in \mathcal S_{AR}^c\), and that
    %
    %
\[\mathbb P^{\boldsymbol{\theta}_n}\left(\boldsymbol{N}_{t} =\boldsymbol{n}_{t} , R_t=j |\boldsymbol{x}_{0:t-1}\right)=\sum\limits_{k\in \mathcal S_{AR}^c}p_{kj}^{\left(n\right)}\mathbb P^{\boldsymbol{\theta}_n}\left( \boldsymbol{N}_{t-1} =\boldsymbol{n}_{t}-\boldsymbol{1} , R_{t-1}=k |\boldsymbol{x}_{0:t-1}\right),\]
with the sum in the denominator being over \(k\in \mathcal S_{AR}^c\) since only when \(k\in \mathcal S_{AR}^c\) is \(\boldsymbol{n}_t-\boldsymbol{1}\in \mathcal S^{(t-1)}\) defined; this completes the proof. 
\end{proof}

\subsection{The M-step}\label{myMstep}
Next, the maximisers, \(\boldsymbol\theta_{n+1} = \argmax\limits_{\boldsymbol\theta \in \Theta} Q\left(\boldsymbol\theta,\boldsymbol\theta_n\right)\), are needed. The maximisers for the parameters of each regime are generally problem specific, but the maximisers for the parameters \(p_{ij}\),  \(i,j \in \mathcal S\), can be derived in general. By the work of \cite{hamilton1990},
\begin{align*}
    p_{ij}^{\left(n+1\right)} = \cfrac{\sum\limits_{t=1}^T \mathbb P^{\boldsymbol \theta_n}\left(R_t = j, R_{t-1} = i | \boldsymbol x_{0:T}\right)}{\sum\limits_{t=1}^T \mathbb P^{\boldsymbol \theta_n}\left(R_{t-1} = i | \boldsymbol x_{0:T}\right)}.
\end{align*}
However, note that to get this analytic update for the \(p_{ij}^{\left(n+1\right)}\) parameters, terms involving \(\pi_j\) in Equation (\ref{eqn: Q function simple MRS}) have been treated as if they are unrelated to \(p_{ij}\), \(i,j\in \mathcal S\). However, this is not true when \(\boldsymbol \pi\) is specified as the stationary distribution of the process \(\{R_t\}\), but holds for other cases, such as when \(\boldsymbol \pi\) is some predetermined distribution, or when \(\boldsymbol \pi\) is specified as a parameter to be inferred. Nonetheless, this simplification is appropriate if we assume that, as the sample size grows, the contribution of terms involving \(R_0\) become insignificant. 

\subsection{Model-specific M-step updates}\label{application:M-step}
In electricity price models, it is common to specify \emph{spike} or \emph{drop} regimes as either shifted-Gamma, shifted-log-normal, or occasionally a Gaussian distribution. Here we derive M-step updates for these regimes.
\begin{corollary}\label{lemma: M-step for i.i.d. processes}
Suppose Regime \(i\) is i.i.d.~\(N\left(\mu_i, \sigma_i^2\right)\). The M-step updates \(\mu_i^{\left(n+1\right)}\) and \(\left(\sigma^{\left(n+1\right)}_i\right)^{2}\), for \(n\geq 0\), are 
\begin{align*}
    \mu_i^{\left(n+1\right)} &= \cfrac{\sum\limits_{t=0}^{T}\mathbb P^{\boldsymbol \theta_n}\left(R_t = i|\boldsymbol x_{0:T}\right)x_t}{\sum\limits_{t=0}^{T}\mathbb P^{\boldsymbol \theta_n}\left(R_t = i|\boldsymbol x_{0:T}\right)},
    \\ \left(\sigma_i^{\left(n+1\right)}\right)^{2} &= \cfrac{\sum\limits_{t=0}^{T}\mathbb P^{\boldsymbol \theta_n}\left(R_t = i|\boldsymbol x_{0:T}\right)\left(x_t-\mu_i^{\left(n+1\right)}\right)^2}{\sum\limits_{t=0}^{T}\mathbb P^{\boldsymbol \theta_n}\left(R_t = i|\boldsymbol x_{0:T}\right)}.
\end{align*}
\end{corollary}

\begin{proof}
The results holds after differentiating the \(Q\left(\cdot,\boldsymbol\theta_n\right)\) and solving for zeros. That \(\mu_i^{\left(n+1\right)}\) is a maximiser is shown by the second derivative test. Furthermore, \(\left(\sigma_i^{\left(n+1\right)}\right)^{2}\) can be shown to be a maximiser by comparing the value of \(Q\left(\boldsymbol\theta,\boldsymbol\theta_n\right)\) when \(\sigma_i^2 = \left(\sigma_i^{\left(n+1\right)}\right)^{2}\) to the value of \(Q\left(\boldsymbol\theta,\boldsymbol\theta_n\right)\) evaluated at any other value of \(\sigma_i^2\), and utilising the inequality \(1-{1}/{y} \leq \log y\).
\end{proof}

    %



At time \(t\), if \(R_t=i\) is a shifted-Gamma regime so that \(X_t -q_i \sim Gam\left(\mu_i, \sigma_i^2\right)\), where \(q_i\) is known, the M-step is not completely analytic. However, the dimension of the maximisation problem can be reduced from 2-dimensional to 1-dimensional via the following corollary. 

\begin{corollary}\label{lemma: M-step for gamma i.i.d. processes}
Suppose Regime \(i\) follows an i.i.d.~shifted-Gamma distribution, that is, if \(X_t\) is from Regime \(i\), then \(\left(X_t-q_i\right)\sim Gam\left(\mu_i, \sigma_i^2\right)\), and suppose the parameter \(q_i\) is known. The M-step update for the scale parameter, \(\left(\sigma_i^{\left(n+1\right)}\right)^2\), as a function of \(\mu_i^{\left(n+1\right)}\), is 
\begin{align*}
    \left(\sigma_i^{\left(n+1\right)}\left(\mu_i^{\left(n+1\right)}\right)\right)^2 &= \mu_i^{\left(n+1\right)}\cfrac{\sum\limits_{t=0}^{T}\mathbb P^{\boldsymbol \theta_n}\left(R_t = i|\boldsymbol x_{0:T}\right)\left(x_t-q_i\right)}{\sum\limits_{t=0}^{T}\mathbb P^{\boldsymbol \theta_n}\left(R_t = i|\boldsymbol x_{0:T}\right)}.
    %
\end{align*}
The update for \(\mu_i\) is then found by finding 
\begin{align*}
    \mu_i^{\left(n+1\right)} = \argmax_{\mu_i \in \left(0,\infty\right)}& \Bigg\{-\mu_i \log \left(\sigma_i^{\left(n+1\right)}\left(\mu_i\right)\right)^2\sum\limits_{t=0}^{T}\mathbb P^{\boldsymbol \theta_n}\left(R_t = i|\boldsymbol x_{0:T}\right) 
    	\\&{}- \log\Gamma\left(\mu_i\right)\sum\limits_{t=0}^{T}\mathbb P^{\boldsymbol \theta_n}\left(R_t = i|\boldsymbol x_{0:T}\right) 
    \\&{} + \left(\mu_i-1\right)\sum\limits_{t=0}^{T}\mathbb P^{\boldsymbol \theta_n}\left(R_t = i|\boldsymbol x_{0:T}\right)\log\left(x_t-q_i\right) 
    	\\&{}- \mu_i\sum\limits_{t=0}^{T}\mathbb P^{\boldsymbol \theta_n}\left(R_t = i|\boldsymbol x_{0:T}\right)\Bigg\} ,
\end{align*}
where \(\Gamma\left(\cdot\right)\) is the Gamma function. 
\end{corollary}

\begin{proof}
The result follows after differentiating \(Q\) with respect to \(\sigma_i^2\), and solving for the stationary point, which is a maximum by the second derivative test. 
\end{proof}

\begin{corollary}\label{lemma: M-step for log-normal i.i.d. processes}
Suppose Regime \(i\) follows i.i.d.~shifted-log-normal dynamics, that is, if \(X_t\) is from Regime \(i\), then \(\log\left(X_t-q_i\right)\sim N\left(\mu_i, \sigma_i^2\right)\), and suppose the parameter \(q_i\) is known. The M-step updates are
\begin{align*}
    \mu_i^{\left(n+1\right)} &= \cfrac{\sum\limits_{t=0}^{T}\mathbb P^{\boldsymbol \theta_n}\left(R_t = i|\boldsymbol x_{0:T}\right)\log\left(x_t-q_i\right)}{\sum\limits_{t=0}^{T}\mathbb P^{\boldsymbol \theta_n}\left(R_t = i|\boldsymbol x_{0:T}\right)},
   \\ \left(\sigma_i^{\left(n+1\right)}\right)^2 &= \cfrac{\sum\limits_{t=0}^{T}\mathbb P^{\boldsymbol \theta_n}\left(R_t = i|\boldsymbol x_{0:T}\right)\left(\log\left(x_t-q_i\right)-\mu_i^{\left(n+1\right)}\right)^2}{\sum\limits_{t=0}^{T}\mathbb P^{\boldsymbol \theta_n}\left(R_t = i|\boldsymbol x_{0:T}\right)}.
\end{align*}

\end{corollary}

\begin{proof}
The proof is similar to the proof of Corollary \ref{lemma: M-step for i.i.d. processes}. 
\end{proof}

Note that Corollaries \ref{lemma: M-step for gamma i.i.d. processes} and  \ref{lemma: M-step for log-normal i.i.d. processes} assume the parameter \(q_i\) is known. This is necessary for the shifted-log-normal distribution \citep{hill1963} and the shifted-Gamma distribution when the shape parameter \(\mu_i\) is less than 1 \citep{johnson1994}. Furthermore, for the shifted-Gamma distribution, \cite{johnson1994} observe that related issues arise when \(\mu\) is near 1, and advise against maximum likelihood estimation of \(q_i\) when \(\mu<2.5\). Simulations suggest this is also good advice when fitting MRS models with shifted-Gamma regimes \citep{lewis2018}. 

In electricity price modelling literature it is common to specify a `base regime' as an AR(1) process. In existing literature the AR(1) regimes are assumed to evolve at every time \(t\) but are only observed when in that regime. Another possibility is that the AR(1) processes evolve only when they are observed, that is, define \(\tau_i\left(t\right) = \sum\limits_{\ell=0}^t\mathbb I\left(R_\ell=i\right)\), and the AR(1) process in Regime~\(i\) as \(\{B_{\tau_i\left(t\right)}\}_{\tau_i\left(t\right)\in \mathbb N}\). Our algorithms are applicable to both specifications; however, for simplicity, here we treat the former only. For more details on the latter specification, see \cite{lewis2018}.

\begin{corollary}
If Regime \(i\) is an AR(1) regime of an MRS model, the M-step of the EM algorithm can be executed as follows. The updates \(\alpha_i^{\left(n+1\right)}\) and \(\left(\sigma_i^{\left(n+1\right)}\right)^2\) as functions of \(\phi_i^{\left(n+1\right)}\) are
\begin{align}
    \alpha_i^{\left(n+1\right)}&\left(\phi_i^{\left(n+1\right)}\right) 
    = \cfrac{\sum\limits_{t=0}^T\sum\limits_{m=1}^{t+1} \mathbb P^{\boldsymbol{\theta}_n}\left(R_t=i,N_{t,i}=m|\boldsymbol x_{0:T}\right)B_{t,m}^{\left(i\right)}}{\sum\limits_{t=0}^T\sum\limits_{m=1}^{t+1} \mathbb P^{\boldsymbol{\theta}_n}\left(R_t=i,N_{t,i}=m|\boldsymbol x_{0:T}\right)A_{t,m}^{\left(i\right)}}, \nonumber
    \\
    \sigma_i^{\left(n+1\right)}&\left(\phi_i^{\left(n+1\right)}\right)^2 
    = \cfrac{\sum\limits_{t=0}^T\sum\limits_{m=1}^{t+1} \mathbb P^{\boldsymbol{\theta}_n}\left(R_t=i,N_{t,i}=m|\boldsymbol x_{0:T}\right) 
	C_{t,m}^{\left(i\right)}
    }{
    \sum\limits_{t=0}^T \mathbb P^{\boldsymbol{\theta}_n}\left( R_t = i |\boldsymbol x_{0:T}\right)}, \nonumber
    \\\text{where } \nonumber
    \\A_{t,m}^{\left(i\right)} &= {\left(\cfrac{1-\left(\phi_i^{\left(n+1\right)}\right)^{m}}{1-\phi_i^{\left(n+1\right)}}\right)}{\left(\cfrac{1+\left(\phi_i^{\left(n+1\right)}\right)}{1+\left(\phi_i^{\left(n+1\right)}\right)^{m}}\right)}, \nonumber
    \\
	B_{t,m}^{\left(i\right)} &= \left(x_t-\left(\phi_i^{\left(n+1\right)}\right)^{m}x_{t-m}\right)\cfrac{1+\left(\phi_i^{\left(n+1\right)}\right)}{1+\left(\phi_i^{\left(n+1\right)}\right)^{m}}, \nonumber
\\        C_{t,m}^{\left(i\right)} &= \cfrac{\left(x_t-\alpha_i^{\left(n+1\right)}\left(\phi_i^{\left(n+1\right)}\right)\left(\cfrac{1-\left(\phi_i^{\left(n+1\right)}\right)^{m}}{1-\phi_i^{\left(n+1\right)}}\right)-\left(\phi_i^{\left(n+1\right)}\right)^{m}x_{t-m}\right)^2}{\left(\cfrac{1-\left(\phi_i^{\left(n+1\right)}\right)^{2m}}{1-\left(\phi_i^{\left(n+1\right)}\right)^{2}}\right)}. \nonumber
\end{align}
The M-step update for \(\phi_i^{\left(n+1\right)}\) is given by
\begin{align*}
    \phi_i^{\left(n+1\right)} & = \argmax_{\phi_i \in \left(-1,1\right)} g\left(\phi_i\right) \\&= \argmax_{\phi_i \in \left(-1,1\right)}  \sum_{t = 0}^T  \sum_{m=1 }^{t+1}  {\mathbb P^{\boldsymbol \theta_n} \left( R_t=i,N_{t,i}=m|\boldsymbol{x}_{0:T}\right) L_{t,m}\left(\phi_i,\sigma_i^{\left(n+1\right)}\right)},
\end{align*}
where
\begin{align*}
    L_{t,m}\left(\phi_i,\sigma_i^{\left(n+1\right)}\right) &= \frac{1}{2}\log\left\{{\cfrac{1-\phi_i^2}{1-\phi_i^{2m}}}\right\} 
    -\log\left\{{\sigma_i^{\left(n+1\right)}\left(\phi_i\right)}\right\}.
\end{align*}
\end{corollary}

\begin{proof}
Differentiate \(Q\) with respect to \(\alpha_i\) and \(\sigma_i^2\) and solve for when the derivative is zero. The parameter \(\alpha_i\) can be shown to be a maximiser by the second derivative test, and \(\sigma_i^2\) can be seen to be a maximiser using the same argument as used in the proof of Corollary \ref{lemma: M-step for i.i.d. processes}. Next, substitute the maximisers, \(\alpha_i^{\left(n+1\right)}\left(\phi_i^{\left(n+1\right)}\right)\) and \(\sigma_i^{\left(n+1\right)}\left(\phi_i^{\left(n+1\right)}\right)^2 \), into (\ref{eqn: Q function simple MRS}) and collect all terms involving \(\phi_i\), to give the function \(g\). That we need to search for the global maximiser of \(g\) on the interval \(\left(-1,1\right)\) only comes from the fact that we have assumed Regime \(i\) is a stationary or mean-reverting process, in which case \(|\phi_i|<1\) is a necessary condition.
\end{proof}

\paragraph{Remark 4.1: Truncation}
The forward and backward algorithms can be computationally costly when \(T\) and/or \(k\) are large, so it may be preferential (or necessary) to truncate the problem. We suggest that the memory of each of the AR(1) processes be truncated. That is, for all \(\boldsymbol n_t\) such that \(n_{t,i}>D-1\) for some \(i\in \mathcal S_{AR}\), we let 
\(f_{X_t|\boldsymbol H_t, \boldsymbol X_{0:t-1}}\left(x_t|\left(\boldsymbol n_t,i\right), \boldsymbol x_{0:t-1}\right) = \widetilde f_{X_t|R_t, \boldsymbol X_{0:t-1}}\left(x_t|i, \boldsymbol x_{0:t-1}\right),\)
where \(\widetilde f_{X_t|R_t, \boldsymbol X_{0:t-1}}\left(x_t|i, \boldsymbol x_{0:t-1}\right)\) is a density that does not depend on \(\boldsymbol n_t\). An appropriate choice of \(\widetilde f\) will be problem-specific. This truncation is equivalent to truncating the state space of \(\boldsymbol H_t\) so that \(N_{t,i}\in \{1,\dots,D\}\), and adjusting the transitions of \(\boldsymbol H_t\) so that the counters remain at \(D\), rather than continuing to increase as they would in the original process. It can be shown that the complexities of the truncated algorithms are \(\mathcal O(M^2D^kTk^k)\). For processes that decay to stationary between times at which they are observed (such as the AR(1) processes in independent-regime MRS models used in this paper, see Section \ref{application:M-step}), then taking \(\widetilde f_{X_t|R_t, \boldsymbol X_{0:t-1}}\left(x_t|i, \boldsymbol x_{0:t-1}\right)\) as the stationary distribution in Regime \(i\) is a logical choice. Furthermore, \(D\) should be chosen large enough so that there is a low probability that \(\{R_t\}\) ever visits any specific state for more than \(D\) consecutive transitions and of course \(D>k\).

\section{Simulation studies}\label{simulations}
We perform a simulation study to examine the properties of the algorithms. The models used in the simulations are the following:
\begin{align}
	X_t &= 
		\begin{cases}
			B_{t} & \text{for }R_t=1,\\
			S_{t} & \text{for }R_t=2,
		\end{cases} \tag{Model 1}
\end{align}
where \(\{B_{t}\}\) is an AR(1) process defined by \(B_{t} = 0+0.75 B_{_{t}-1} + \varepsilon_{t}\), \linebreak\(\varepsilon_{t}\sim\) i.i.d.~\(N(0,1)\), \(S_{t}\sim\) i.i.d~N\(\left(0,1\right)\) and \(\{R_t\}\) is a Markov chain with state space \(\mathcal S=\{1,2\}\), transition matrix entries \(p_{11}=p_{22}=0.9\), and initial distribution \(\left(0.5,0.5\right)\); 
\begin{align}
	X_t &= 
		\begin{cases}
			B_{t}^{\left(1\right)} & \text{for }R_t=1,\\
			B_{t}^{\left(2\right)} & \text{for }R_t=2,
		\end{cases} \tag{Model 2}
\end{align}
where \(\{B_{t}^{\left(i\right)}\}\) are independent AR(1) process defined by \(B_{t}^{\left(1\right)} = 0+0.9 B_{t-1}^{\left(1\right)} + \varepsilon^{\left(1\right)}_{t}\) and \(B_{t}^{\left(2\right)} = 0+0.4 B_{t-1}^{\left(2\right)} + \varepsilon_{t}^{\left(2\right)}\), \(\varepsilon_{t}^{\left(i\right)}\sim\) i.i.d.~\(N(0,1)\) \(i=1,2\), and \(\{R_t\}\) is a Markov chain with state space \(\mathcal S=\{1,2\}\), transition matrix entries \(p_{11}=p_{22}=0.6\), and initial distribution \(\left(0.5,0.5\right)\).  

\begin{figure}[htbp]
\begin{center}
\includegraphics[width=\textwidth]{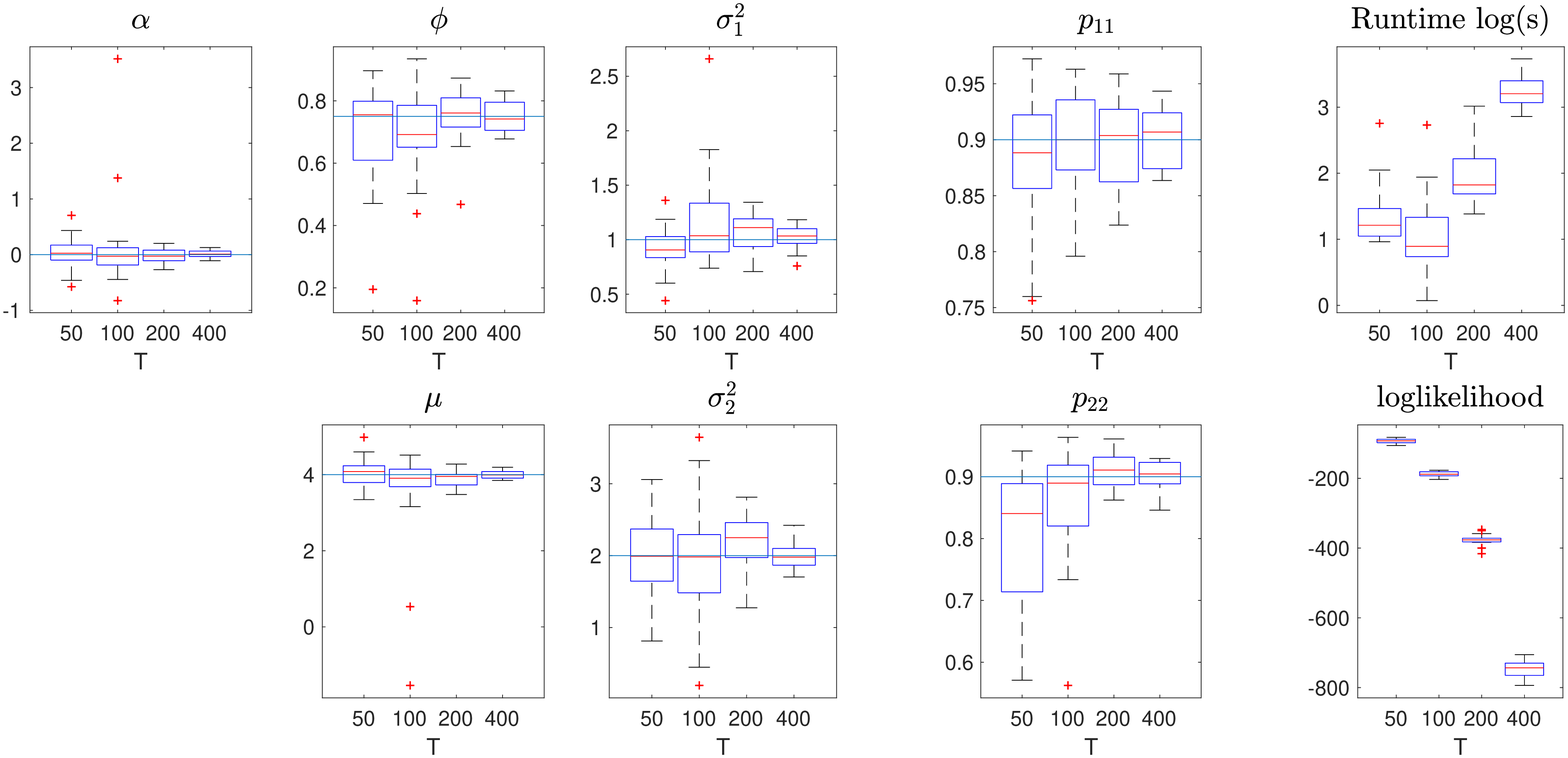}%
\vspace{0.1cm}%
\hrule%
\vspace{0.1cm}%
\includegraphics[width=\textwidth]{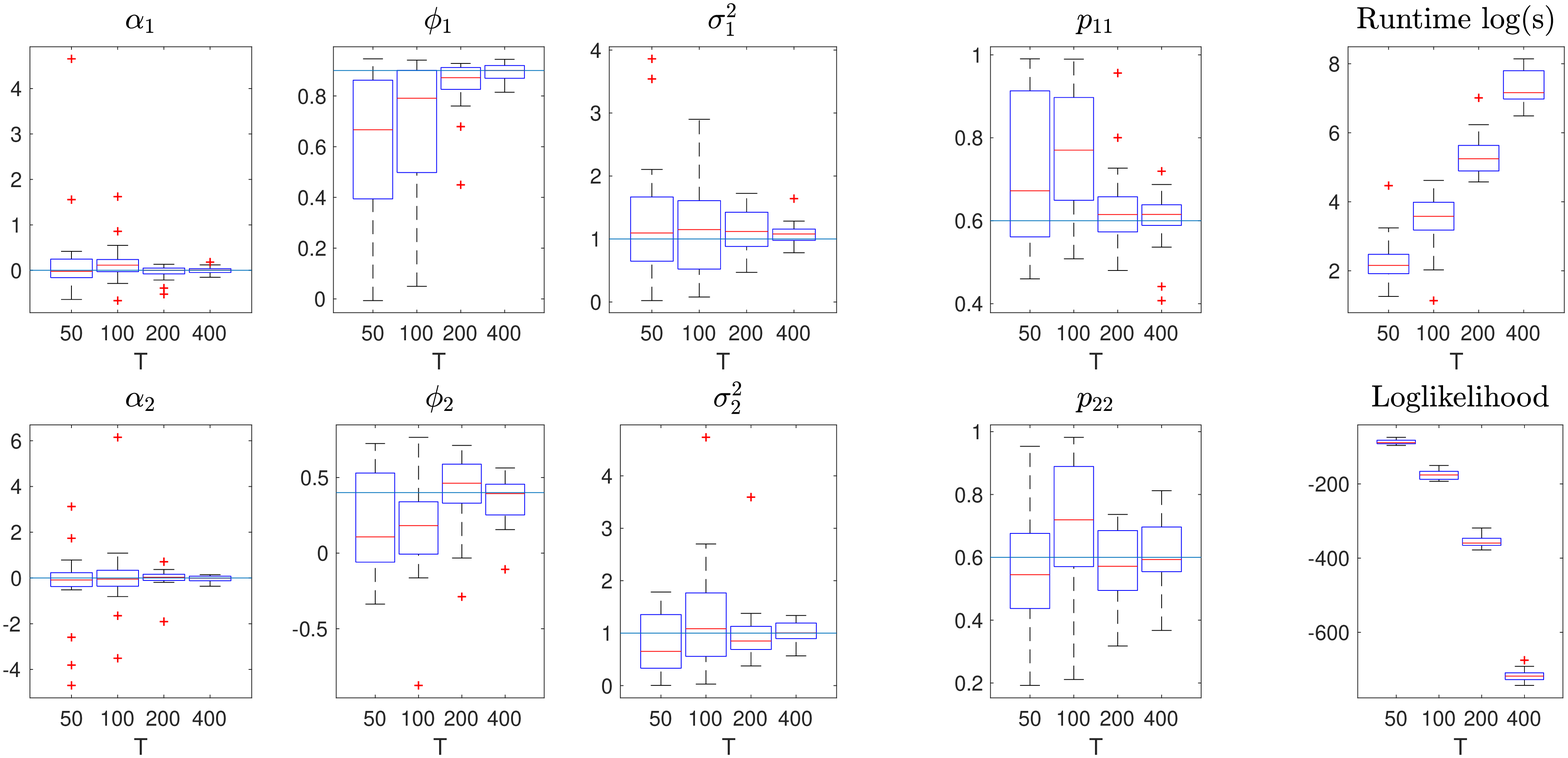}
\caption{Box plots of MLEs, runtime (in log-seconds) and loglikelihood for Models 1 (top) and 2 (bottom) found via the EM algorithm. Each boxplot contains 20 independently simulated data sets. The x-axis is the length of the simulated data set, \(T=50,100,200\) and \(400\). The true parameter values are marked as horizontal lines. For both models the MLE appears to be consistent. For Model 2 (bottom), the two regimes have similar behaviour with only the parameters \(\phi_i\) differing between them by 0.25. With a data set of length \(T=400\) the algorithm appears to be able to differentiate the two regimes.}
\label{fig:bias}
\end{center}
\end{figure}%
To investigate the bias and consistency of the MLE, 20 independent realisations of Models 1 and 2 were simulated for \(T=50,100,200\) and \(400\), and the EM algorithm used to find the MLE. The terminating criteria for the algorithms was to stop when either the increase in the likelihood, or the step-size, as measure by \mbox{\(|\boldsymbol{\theta}_{n+1} - \boldsymbol{\theta}_{n}|_\infty\)} was less that \(1.5\times10^-8\). To attempt to avoid local maxima, the EM algorithm was initialised from 100 randomised values centred around the true parameters. Of the corresponding 100 terminating points of the EM algorithm, the parameters that achieved the highest loglikelihood value were kept. Figure~\ref{fig:bias} shows box plots of the 20 terminating points of the EM algorithm, one for each simulation, as well as the log-runtime and value of the loglikelihood. For both models the MLE appears to be converging to the true parameter value as sample sizes increase. Generally, there appears to be a much larger variation in the MLEs for Model 2 than for Model 1. This could be because the inference problem for Model 2 is harder, as the regimes in Model 2 are more similar than they are in Model 1, or because of the nature of the hidden Markov chain is such that there is a lower probability of remaining in each regime (\(p_{11}=p_{22}=0.6\) for Model 2, compared to \(p_{11}=p_{22}=0.9\) for Model 1), or both. Regarding the run time, for data sets of length 100 and greater, the empirical results suggest the complexity is approximately \(\mathcal O\left(T^{1.02}\right)\) for Model 1 and \(\mathcal O\left(T^{1.70}\right)\) for Model 2, which agrees with, and is significantly better than, our theoretical upper bound.

\begin{figure}[htbp]
\begin{center}
\includegraphics[width=\textwidth]{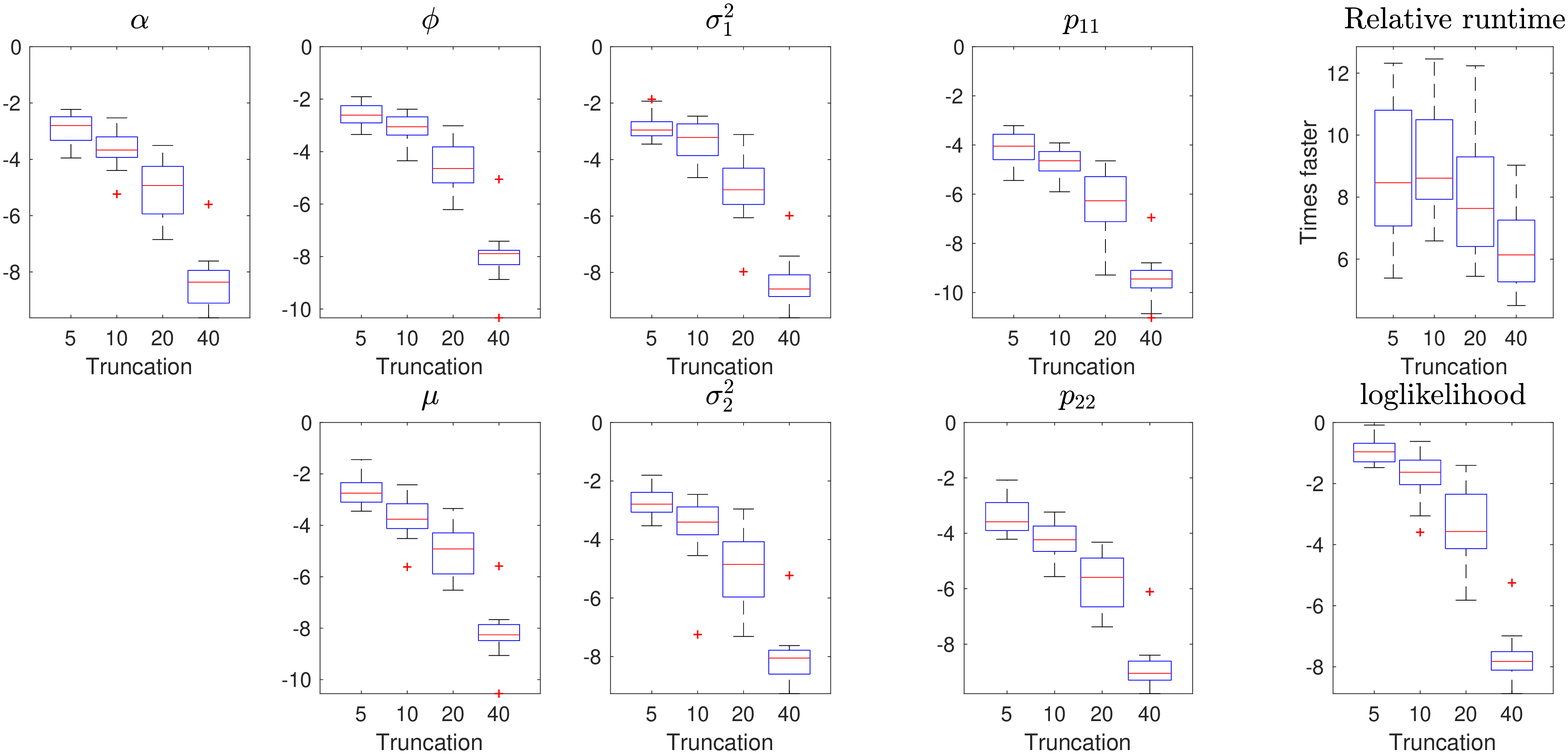}%
\vspace{0.1cm}%
\hrule%
\vspace{0.1cm}%
\includegraphics[width=\textwidth]{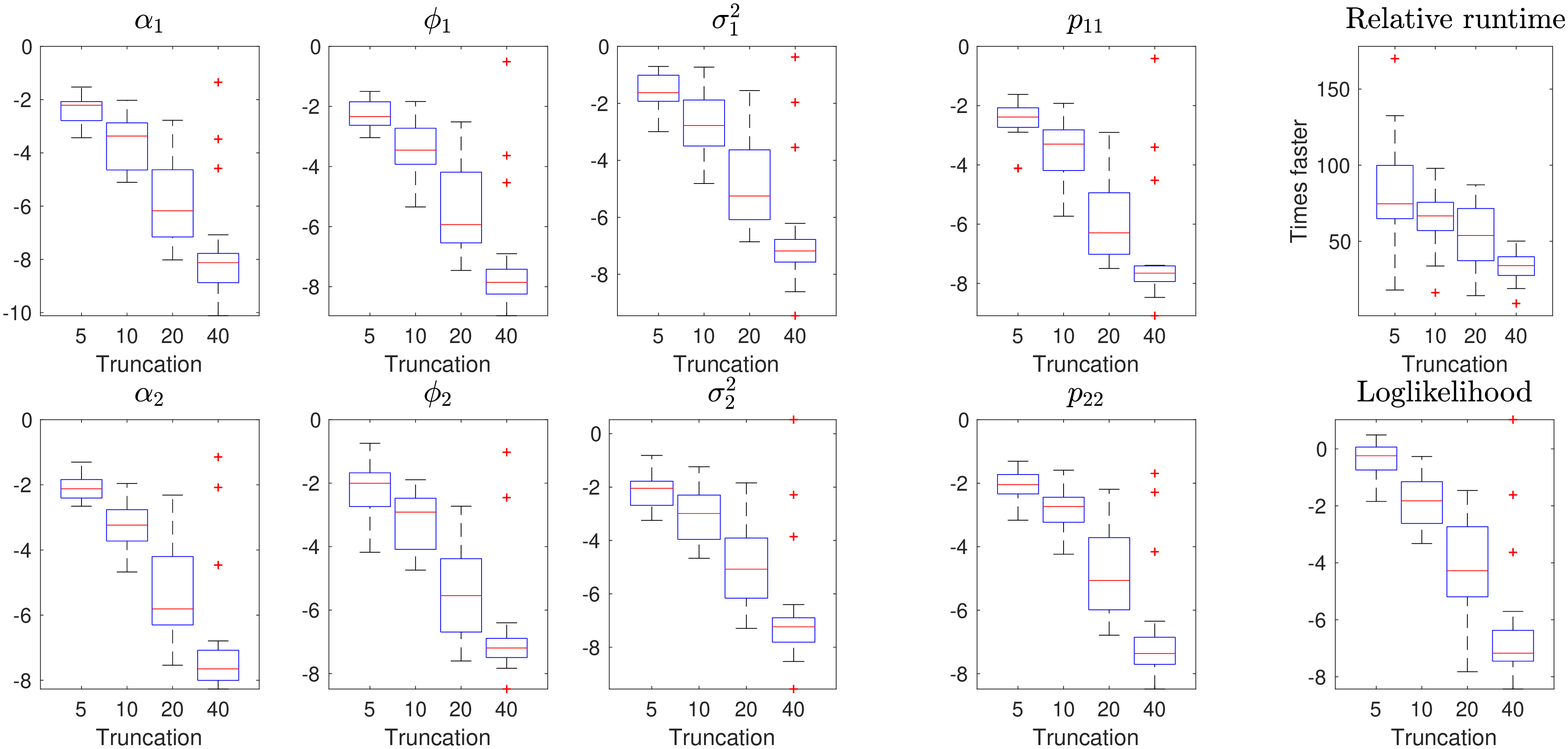}
\caption{Box plots of the \(\log_{10}\) of the absolute value of the difference between the MLEs and loglikelihood values found by the full algorithm and the truncated algorithm, and plots of the ratio of the runtime of the truncated algorithm compared to the full algorithm, for Model 1 (top) and Model 2 (bottom). Each boxplot contains 20 independently simulated data sets. The x-axis is the maximum memory of the counters, \(D=5,10,20,40\). For plots corresponding to parameter estimates or loglikelihood values, the y-axis is the \(\log_{10}\) error between the MLE found by the full algorithm and the truncated algorithm, for the runtime plots the y-axis is multiples of the runtime of the full algorithm. For both models, a truncation level of \(D=40\) gives a relative error of the order approximately \(10^{-8}\) which is the order of the stopping criterion of the EM algorithm.}
\label{fig:trunc}
\end{center}
\end{figure}

To investigate the truncation method, we used the same data sets simulated above for \(T=400\) and applied our truncated algorithm to this data with various truncation levels, \(D=5,10,20,40\). The termination criteria for the EM algorithm was to stop when either the increase in the likelihood, or the step-size, as measure by \mbox{\(|\boldsymbol{\theta}_{n+1} - \boldsymbol{\theta}_{n}|_\infty\)} was less that \(1.5\times10^{-8}\). Figure~\ref{fig:trunc} plots the \(\log_{10}\) of the absolute value of the difference between the parameters recovered by the truncated algorithm and the full algorithm. Figure~\ref{fig:trunc} also shows \(\log_{10}\) of the absolute value of the difference in the loglikelihoods achieved by the truncated algorithm and the full algorithm. Figure  \ref{fig:trunc} also plots of the relative runtime of the truncated algorithm compared to the full algorithm, that is, the ratio of the time taken for the truncated algorithm compared to the full algorithm. The run times for the truncated algorithms are significantly lower than the full algorithm. For these data sets the truncated algorithm performs reasonably well, even when limiting the memory of the counters to just 5 time steps. For a truncation level of 40, the errors are of the order \(10^{-6}\) to \(10^{-10}\) for both models, which is similar to the stopping criteria for the EM algorithm which is of the order \(10^{-8}\). For Model 1 the truncation appears to have a more significant negative effect on the parameter estimates compared to Model 2. This could be because the probability of staying in each regime is higher in Model 1 than in Model 2, and therefore the probability of remaining in one regime for more that \(D\) consecutive transitions is lower in Model 2. It is likely that the error due to truncation of the algorithm is insignificant compared to statistical error of parameter estimates. %

To investigate the convergence properties of the EM algorithm we used five of the simulated data from above with \(T=400\) and used the truncated EM algorithm, with a truncation level of \(D=40\) to search for maxima. For each of the five simulated datasets we ran the EM algorithm 50 times, sampling initial values for the algorithm independently each time and used the same terminating criteria as before. The sampling distributions for the initial parameter values of the algorithm are summarised in Table \ref{table: init distns}. For Model 2 we have ordered the terminating values of the EM algorithm such that \(\phi_1>\phi_2\) for identifiability. 

Observing Figure \ref{fig:randstart} we see that, for most of the simulations and most of the starting values, the EM algorithm finds a single maxima. In Figure \ref{fig:randstart} the black cross represents the point which achieved the highest maximum. We refer to this point as the optimal parameter value. The proportion of times that the algorithm converged to the optimal parameter value is reported in Table \ref{table: nnear}. Simulation 4 of Model 1 is outstanding since only 14\% of initial values resulted in the algorithm finding the optimal value. In Figure \ref{fig:randstart} there are some instances when the algorithm terminates at a suboptimal point; in particular simulated datasets 3 and 4 for Model 1. As shown in. Figure \ref{fig:randstart}, for Model 1 the optimal parameter set appears to reasonably estimate the true parameters for all simulated dataset. 

Estimating Model 2 is much harder since the regimes are both very similar. Figure \ref{fig:randstart} shows that the optimal parameter set estimates the true parameters reasonably for simulated datasets 2-5, but not for simulated dataset 1. The behaviour of the algorithm for simulated dataset 1 for Model 2 is particularly interesting. For this dataset, the algorithm terminates at one of two distinct locations. One of these terminating points is at \(p_{11}=1\) and hence Regime 1 is absorbing. This means the algorithm has converged to a point where all but the first observation are captured by Regime 1. As \(p_{11}\) tends to \(1\) the algorithm is able to send \(\sigma_2^2\to0\) to achieve arbitrarily large values of the likelihood and the model is unidentifiable. To prevent this behaviour, we suggest that the parameters \(\sigma_i^2\) and \(p_{ij}\) are restricted so that they are away from the boundary. Indeed, for simulated dataset 1 of Model 2, if we take the terminating values of the algorithm which lie away from the boundary, then we get reasonable estimates of the true parameters. For more discussion see \cite{lewis2018}.
\begin{table}[bt]
\centering
\caption{Distributions used to sample initial parameter values for the EM algorithm.}\label{table: init distns}
\begin{tabular}{l|cccccc}
Parameter & \(\alpha_i\) & \(\phi_i\) & \(\sigma_i^2\) & \(\mu\) & \(p_{11}\) & \(p_{22}\)\\\hline
Distribution & \(U(-1,1)\) & \(U(-1,1)\) & \(U(0,4)\) & \(U(0,8)\) & \(U(0,1)\) & \(U(0,1)\)
\end{tabular}
\end{table}
\begin{figure}[htbp]
\begin{center}
\includegraphics[width=\textwidth, trim = {120 0 120 0}, clip]{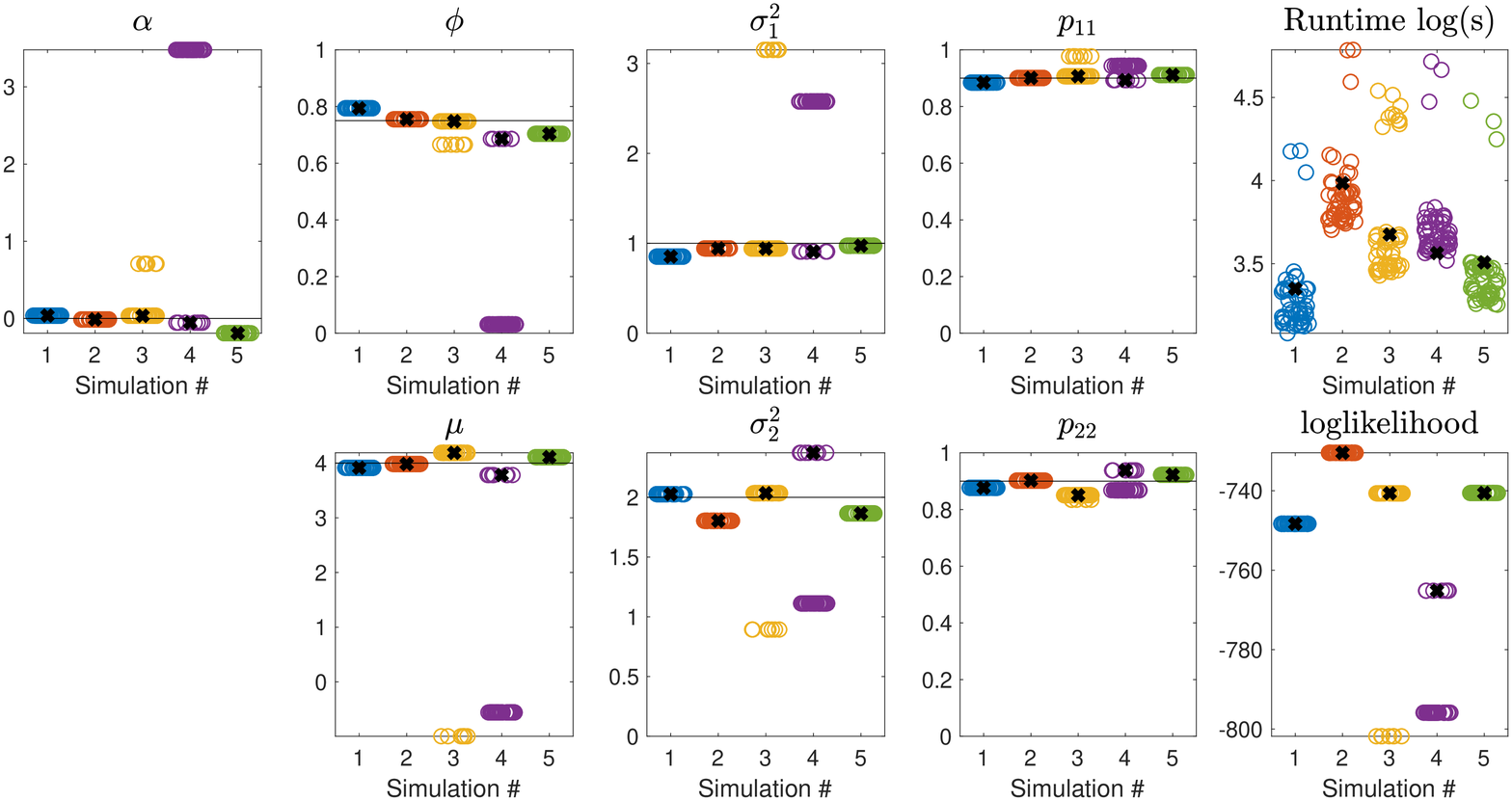}%
\vspace{0.1cm}%
\hrule%
\vspace{0.1cm}%
\includegraphics[width=\textwidth, trim = {120 0 120 0}, clip]{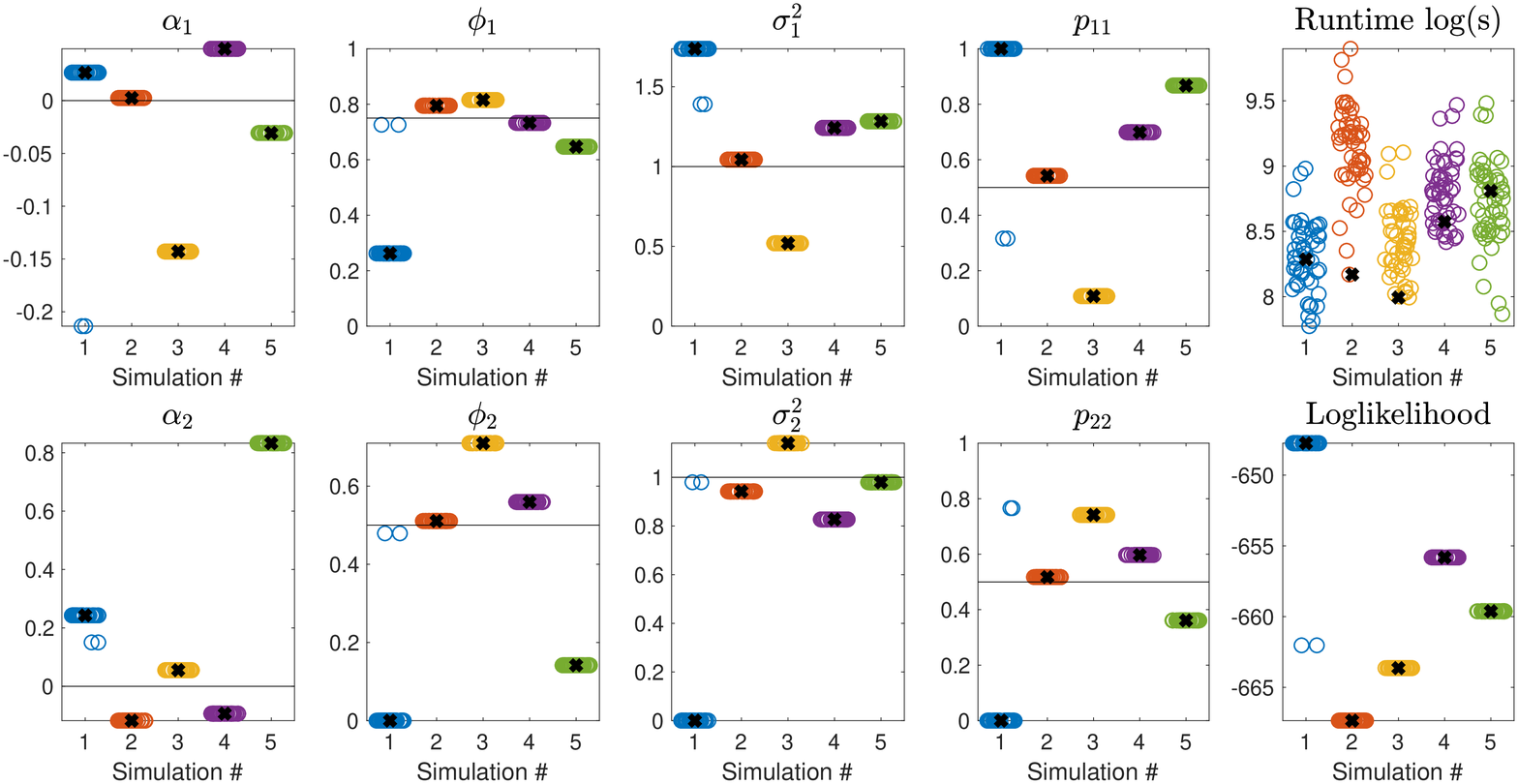}
\caption{Scatter plots of terminating values of the truncated version of the EM algorithm with truncation parameter \(D=40\) for Model 1 (top) and Model 2 (bottom), given random initial values sampled according to the distributions in Table \ref{table: init distns}. For each model 5 datasets of length \(T=400\) were simulated, as indicated on the \(x\)-axis. For each simulation the truncated EM algorithm was run 50 times, initialising the algorithm at independently sampled random initial values each time. The algorithm was run until either the increase in the loglikelihood, or the difference between successive parameter values was less than \(1.5\times10^{-8}\) and the terminating values recorded. The black cross corresponds to values associated with the highest loglikelihood value. Each circle corresponds to a terminating value of the EM algorithm. The points have been `jittered' so they do not all lie on top of each other. }
\label{fig:randstart}
\end{center}
\end{figure}
\begin{table}[bt]
\centering
\caption{Proportion of terminating values which lie within 0.0001 of the optimal parameter value, as measured by the \emph{sup norm}, \(|\cdot|_\infty\).}\label{table: nnear}
\begin{tabular}{l|ccccc}
Simulation & 1 & 2 & 3 & 4 & 5\\\hline
Model 1 & 1 & 1 & 0.86 & 0.14 & 1 \\\hline
Model 2 & 0.96 & 1 & 1 & 1 & 1
\end{tabular}
\end{table}

\section{An application to South Australian wholesale electricity market}\label{application}
The dataset consists of 81,792 half-hourly spot prices from the South Australian electricity market (available at the AEMO website \citep{aemo}) for the period  00:00 hours, \(1^{\text{st}}\) of January 2013, to 23:30 hours, \(31^{\text{st}}\) of September 2017. Note that this dataset contains a period of 14 days over which the market was suspended from 4:00pm, on the \(28^{\text{th}}\) of September until 10:30pm on the \(11^{\text{th}}\) of October. During this period prices were set by AEMO. We ignore this fact in our modelling and include them in the data set anyway.

Following a common practice in the literature we model daily average prices and thus we have a dataset of 1,704 daily average price observations to which we fit our model. The data that we model is plotted in Figure~\ref{fig:classification}.
To model the South Australian wholesale electricity market, we break the price process up in to two components, \(P_t = S_t+X_t\), where \(P_t\) is the price on day \(t\), \(S_t\) is a deterministic trend component, and \(X_t\) is a stochastic component which is to be modelled by an independent regime MRS model. 

Electricity spot prices exhibit seasonality on daily, weekly, and longer scales. To capture this multi-scale seasonality, the trend component consists of two parts: a short-term component, \(g_t\), and a long-term component, \(h_t\), so \(S_t=g_t+h_t\). We model the long-term component, \(h_t\), using wavelet filtering since it has been shown to perform well for this application \citep{janczura2013b}, and use Daubechies 24 wavelets and a level 6 approximation \citep{janczura2013b}. 
We use the short-term component, \(g_t\), to capture the mean price for different days of the week and indicator functions to model this: \begin{align*}g_t &= \beta_\text{Mon}\mathbb I \left(t \in \text{Mon}\right) + \beta_\text{Tue}\mathbb I \left(t \in \text{Tue}\right) + \dots+ \beta_\text{Sun}\mathbb I \left(t \in \text{Sun}\right),\end{align*}
where \(\beta_\text{Mon}\), \(\beta_\text{Tue},\) \dots , \(\beta_\text{Sun}\) are the mean deviations from the long-term trend price on Monday, Tuesday, \dots , Sunday, respectively. 

Following \cite{janczura2013b} we use the RFP (recursive filter on prices) method to estimate the seasonal component in the presence of extreme observations. The method first uses the raw price series to estimate the trend model, then removes this from the data. Next, the standard deviation of these altered prices is estimated, and any prices that are more than three standard deviations from the current estimate of the trend are replaced with the value of the estimated trend at that point. The procedure then re-estimates the trend component on the original data set with spikes removed. 

We consider the following four models for the stochastic component: 
\begin{align}
	X_t = 
	\begin{cases}
		B_{t} & \mbox{if } R_t = 1, \\
		S_t & \mbox{if } R_t = 2,
	\end{cases}\tag{M 1}
\end{align}
where \(B_t=\alpha + \phi B_{t-1} + \sigma_1 \varepsilon_t\) is an AR(1) process with \(\varepsilon_t\sim\) i.i.d.~N\(\left(0,1\right)\), and \(S_t-q_3\), where \(q_3\) is a shifting parameter, follows either a Gamma distribution (M1-Gamma), or a log-normal distribution with parameters \(\mu_2\) and \(\sigma_2^2\) (M1-LN). The last two models introduce a `drop' regime as well: 
\begin{align}
	X_t = 
	\begin{cases}
		B_{t} & \mbox{if } R_t = 1, \\
		S_t & \mbox{if } R_t = 2, \\
		D_t & \mbox{if } R_t = 3,
	\end{cases}\tag{M 2}
\end{align}
where \(B_{t}\) and \(S_t\) are as above, and \(-D_t+q_1\), where \(q_1\) is a shifting parameter, follows a log-normal distribution with parameters \(\mu_3\) and \(\sigma_3^2\). The use of the shifting parameters \(q_1\) and \(q_3\) was proposed by \cite{janczura2009}. As previosuly mentioned, estimating the shifting parameters of these distribution is known to be a difficult task \citep{johnson1994}, so we fix \(q_1\) and \(q_3\) as the first and third quantiles of the detrended data, as suggested by \cite{janczura2012}.

The models were fitted to the detrended data using the truncated EM algorithm with truncation level \(D = 56\) (8 weeks). The Bayesian Information Criterion (BIC) values of these models are reported in Table~\ref{table:BIC}. From the BIC values we choose model M1-LN, with an AR(1) base regime and a log-normal spike regime. The parameters for this model are reported in Table~\ref{table:parameters}.  Of course, for a rigorous treatment of this modelling problem, model assumptions should be checked and we should not rely solely on the BIC. Using the smoothed probabilities obtained while fitting model M1-LN, the prices can be classified into which regime is most likely. This is shown in Figure~\ref{fig:classification}, where prices are highlighted in red if \(\mathbb P^{\widehat{\boldsymbol \theta}}\left(R_t=2|\boldsymbol x_{0:T}\right)>0.5\), where \(\widehat{\boldsymbol\theta}\) is the MLE.

\begin{table}[bt]
\centering
\caption{BIC values of models M1-LN, M1-Gamma, M2-LN, and M2-Gamma, fitted to the SA electricity market data.}\label{table:BIC}
\begin{tabular}{l|cccc}
{Model} & {M1-LN} & {M1-Gamma} & {M2-LN} & {M2-Gamma}\\\hline
{BIC}     &15572     &   15681          &   15575   &   15706	   \\   
\end{tabular}
\end{table}

\begin{table}[bt]
\centering
\caption{Parameter estimates of model M1-LN, which has the lowest BIC, fitted to the SA electricity market data.}\label{table:parameters}
 \begin{tabular}{c|ccccc} 
			\(B_t\) & \(\alpha\) & \(\phi\) & \(\sigma_1^2\) & & \(p_{11}\) \\ 
			& -3.257 & 0.6830 & 213.26 && 0.9140  \\ \hline
			\(S_t\) & \(q_3\) & \(\mu_2\) & \(\sigma_2^2\)& & \(p_{22}\) \\ 
			& 7.106 & 3.751 & 1.268 & & 0.3945  \\ 
\end{tabular}
\end{table}

\begin{figure}[htbp]
\begin{center}
\includegraphics[width = 1\textwidth]{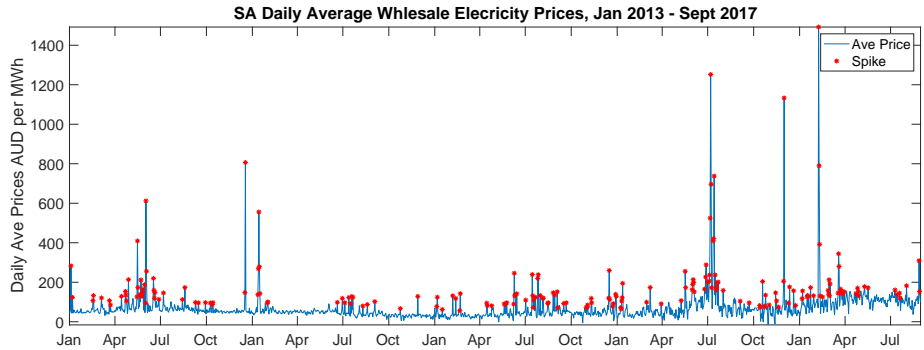}
\caption{SA daily average wholesale electricity price for the period 1-Jan-2013 until 30-Sept-2017. Prices that have greater probability of being from Regime 2 in Model M1-LN are highlighted in red. }
\label{fig:classification}
\end{center}
\end{figure}

\section{Conclusions}\label{conclusion}
In this paper we have developed novel techniques for independent-regime MRS models. Specifically, we consider models that are a collection of independent AR(1) processes, where only one process is observed at each time \(t\), and which regime is observed is determined by a hidden Markov chain. We develop forward, backward and EM algorithms for these models, and show that the methods we develop here can outperform the existing method of inference used in the electricity price modelling literature, the EM-like algorithm \citep{janczura2012}. 

The construction of these methods relies on the idea of augmenting the hidden Markov chain with a set of counters, which keep track of the number of transitions since the last visit to AR(1) regimes. The forward algorithm can be used to evaluate the likelihood and filtered and prediction probabilities, which are used as inputs to the backward algorithm. The backward algorithm is used to evaluate the smoothed probabilities. Together, the forward-backward procedure executes the E-step of the EM algorithm. We showed that the complexity of the forward and backward algorithms is \(\mathcal O\left(M^2T^{k+1}k^k\right)\), where \(M\) is the number of regimes in the model, \(T\) the length of the observed sequence, and \(k\) the number of AR(1) regimes in the model. This complexity may be impractically large if \(T\) or \(k\) are large, so we introduce an approximation where the memory of AR(1) processes in the model is truncated. These truncated methods are \(\mathcal O\left(M^2D^kTk^k\right)\), where \(D\) is the memory of the counters. 

Simulation suggests that the MLE found via the EM algorithm is consistent, and that, even for models with regimes with similar characteristics (such as Model 2), the MLE is still a reasonable estimator when the sample size is large enough (400 observations in this case). Simulations also suggest that the truncation method is a reasonable approximation to the full likelihood, while improving runtime and memory requirements significantly. As is typical for hill-climbing algorithms, there is the possibility of the algorithm terminating at sub-optimal values depending on the initial values of the algorithm. We explored some of this behaviour via a simulation study and found, at most, two local maxima of the likelihood function. Of course, the behaviour of the algorithm is going to be model- and data-specific. This simulation study also highlighted possible identifiability issues which can arise. However, these can be rectified by restricting parameters away from the boundary. 

Lastly, we apply our methods to estimate four models for the South Australian wholesale electricity market and find that a 2-regime model with an AR(1) base regime and shifted log-normal spike regime is best as measured by the BIC. We also demonstrate how prices can be classified into regimes using the smoothed probabilities.



\bibliography{Bibliography}

\end{document}